\documentclass{article}

\usepackage{microtype}
\usepackage{graphicx}
\usepackage{booktabs} 

\usepackage{hyperref}

\usepackage[accepted]{icml2020}

\usepackage[]{algorithm}

\icmltitlerunning{Private Counting from Anonymous Messages}

\renewcommand{\algorithmicrequire}{\textbf{Input:}}
\renewcommand{\algorithmicensure}{\textbf{Output:}}
\newcommand{\Procedure}[2]{\STATE {{\bf procedure} {\sc #1}}} 
\newcommand{\State}{\STATE \hspace*{1em}}
\newcommand{\Return}{{\bf return \/}}
\newcommand{\EndProcedure}{}

\usepackage{amsmath}
\usepackage{amssymb}
\usepackage{amsthm}
\usepackage{graphicx}
\usepackage{thm-restate}
\usepackage[normalem]{ulem}
\usepackage{balance}

\usepackage{enumitem}
\newcommand{\nc}{\newcommand}
\newcommand{\DMO}{\DeclareMathOperator}

\allowdisplaybreaks

\newcommand{\poly}{\text{poly}}
\nc{\MS}{\mathcal{S}}
\nc{\MP}{\mathcal{P}}
\nc{\MR}{\mathcal{R}}
\nc{\cM}{\mathcal{M}}

\nc{\MZ}{\mathcal{Z}}
\newcommand{\cX}{\mathcal{X}}
\DMO{\Binom}{Binom}
\newcommand{\E}{\mathbb{E}}
\DMO{\Var}{Var}
\newcommand{\cA}{\mathcal{A}}
\newcommand{\ba}{\mathbf{a}}

\newcommand{\bo}{\mathbf{o}}

\newcommand{\bx}{\mathbf{x}}
\newcommand{\by}{\mathbf{y}}

\newcommand{\R}{\mathbb{R}}
\newcommand{\N}{\mathbb{N}}

\newcommand{\diff}{\text{diff}}
\newcommand{\hZ}{\hat{Z}}
\newcommand{\Zz}{\Z_{\geq 0}}
\nc{\BN}{\mathbb{N}}
\newcommand{\Z}{\mathbb{Z}}
\nc{\BZ}{\mathbb{Z}}
\newcommand{\cR}{\mathcal{R}}

\newcommand{\bone}{\mathbf{1}}
\newcommand{\eps}{\varepsilon}
\renewcommand{\epsilon}{\varepsilon}
\nc{\hy}{\hat{y}}

\nc{\ep}{\eps}
\newcommand{\cD}{\mathcal{D}}

\DeclareMathOperator{\supp}{supp}

\DeclareMathOperator{\DLap}{DLap}
\DeclareMathOperator{\Bin}{Bin}
\DeclareMathOperator{\Ber}{Ber}

\DeclareMathOperator{\Poi}{Poi}
\DeclareMathOperator{\NB}{NB}
\DeclareMathOperator{\DCP}{DCP}

\newtheorem{theorem}{Theorem}
\newtheorem{observation}[theorem]{Observation}
\newtheorem{lemma}[theorem]{Lemma}

\newtheorem{definition}[theorem]{Definition}

\newtheorem{corollary}[theorem]{Corollary}
\newtheorem{remark}[theorem]{Remark}

\begin{document}

\twocolumn[
\icmltitle{Private Counting from Anonymous Messages: \\
Near-Optimal Accuracy with Vanishing Communication Overhead
}



\icmlsetsymbol{equal}{*}

\begin{icmlauthorlist}
\icmlauthor{Badih Ghazi}{googlemtv}
\icmlauthor{Ravi Kumar}{googlemtv}
\icmlauthor{Pasin Manurangsi}{googlemtv}
\icmlauthor{Rasmus Pagh}{ituc,googlemtv}
\end{icmlauthorlist}

\icmlaffiliation{googlemtv}{Google Research, Mountain View}

\icmlaffiliation{ituc}{IT University of Copenhagen}

\icmlcorrespondingauthor{Badih Ghazi}{badihghazi@gmail.com}
\icmlcorrespondingauthor{Ravi Kumar}{ravi.k53@gmail.com}
\icmlcorrespondingauthor{Pasin Manurangsi}{pasin@google.com}
\icmlcorrespondingauthor{Rasmus Pagh}{pagh@itu.dk}

\icmlkeywords{Differential Privacy, Summation, Communication, Shuffled Model, Encode--Shuffle--Analyze}
\vskip 0.3in
]

\printAffiliationsAndNotice{\icmlEqualContribution}

\begin{abstract}


Differential privacy (DP) is a formal notion for quantifying the privacy loss of algorithms.  Algorithms in the central model of DP achieve high accuracy but make the strongest trust assumptions whereas those in the local DP model make the weakest trust assumptions but incur substantial accuracy loss. The shuffled DP model~\cite{bittau17, erlingsson2019amplification, CheuSUZZ19} has recently emerged as a feasible middle ground between the central and local models, providing stronger trust assumptions than the former while promising higher accuracies than the latter.
In this paper, we obtain practical communication-efficient algorithms in the shuffled DP model for two basic aggregation primitives used in machine learning: 1) binary summation, and 2) histograms over a moderate number of buckets.  Our algorithms achieve accuracy that is arbitrarily close to that of central DP algorithms with an expected communication per user essentially matching what is needed without any privacy constraints! 
We demonstrate the practicality of our algorithms by experimentally 
comparing their performance to several widely-used protocols such as Randomized Response~\cite{warner1965randomized} and RAPPOR~\cite{erlingsson2014rappor}.

\end{abstract}

\section{Introduction}

Motivated by the need for scalable, distributed privacy-preserving machine learning, there has been an intense interest, both in academia and industry, on designing algorithms with low communication overhead and high accuracy while protecting potentially sensitive, user-specific information. While many notions of privacy have been  proposed, \emph{differential privacy} (DP)~\cite{dwork2006calibrating,dwork2006our} has become by far the most popular and well-studied candidate, leading to several real-world deployments at companies such as Google~\cite{erlingsson2014rappor,CNET2014Google}, Apple~\cite{greenberg2016apple,dp2017learning}, and Microsoft~\cite{ding2017collecting}, and in government agencies such as the U.S.~Census Bureau \cite{abowd2018us}. Most research has focused on the \emph{central} model of DP where a curator, who sees the raw user data, is required to release a private data structure. 
While many accurate DP  algorithms have been discovered in this framework, the requirement that the curator observes the raw data constitutes a significant obstacle to deployment in many industrial settings where the users do not necessarily trust the central authority. To circumvent this limitation, several works have studied the \emph{local} model of DP~\cite{kasiviswanathan2008what} (also~\cite{warner1965randomized}), which enforces the more stringent constraint that each message sent from a user device to the server is  private. While requiring near-minimal trust assumptions, the local model turns out to inherently suffer from large estimation errors. For numerous basic tasks, including binary summation and histograms that we study in this work, errors are at least on the order of $\sqrt{n}$, where $n$ is the number of users~\cite{beimel2008distributed,ChanSS12}.

{\bf Shuffled Privacy Model.} The shuffled (aka.~anonymous) model of privacy has recently generated significant interest as a potential compromise between the central and local frameworks: having trust assumptions better than the former but enabling estimation accuracies higher than the latter. While the shuffled model was originally studied in the field of cryptography by~\citet{ishai2006cryptography} in their work on cryptography from anonymity, it was first suggested as a framework for privacy-preserving computations by~\citet{bittau17} in their Encode-Shuffle-Analyze architecture. 
This setting only requires the multiset of \emph{anonymized} messages that are transmitted by the different users to be private. Equivalently, this corresponds to the setup where a trusted shuffler \emph{randomly permutes} all incoming messages from the users before passing them to the analyzer. This is illustrated in Figure~\ref{fig:esa}.
We point out that several efficient cryptographic implementations of the shuffler have been considered including mixnets, onion routing, secure hardware, and third-party servers (see, e.g., \cite{ishai2006cryptography, bittau17} for more details). As in all previous work on the shuffled model, we  treat the shuffler as a black box.

\begin{figure}[h]
\centering
\includegraphics[width=0.5\textwidth]{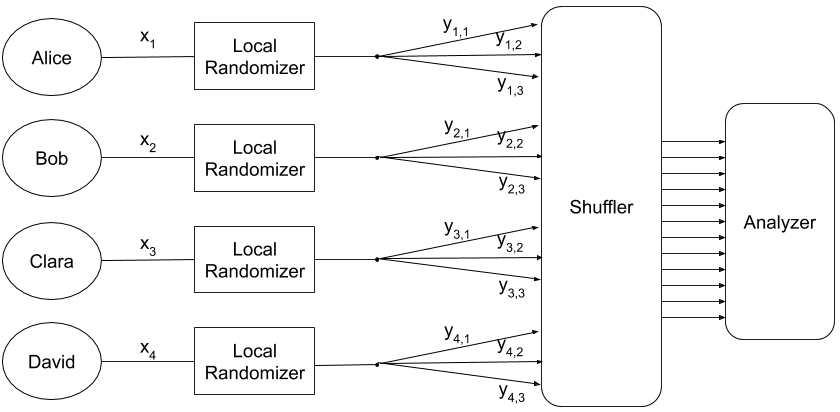}
\caption{The Shuffled Model.}
\label{fig:esa}
\end{figure}

DP in the shuffled model was first formally investigated, independently, by~\citet{erlingsson2019amplification} and~\citet{CheuSUZZ19}. Several recent works have aimed to determine the optimal trade-offs between communication, accuracy, and privacy in this model for various algorithmic tasks~\cite{BalleBGN19,ghazi2019scalable,anon-power, ghazi2019private, balle_merged, balcer2019separating}.

{\bf Summation.} One of the most basic distributed computation problems is \emph{summation} (aka.~aggregation) where the goal of the analyzer is to estimate the sum of the user inputs. In machine learning, and specifically in the nascent field of federated learning \cite{konevcny2016federated} (see, e.g., \cite{kairouz2019advances} for a recent survey), private summation enables private Stochastic Gradient Descent (SGD), which in turn allows the private training of deep neural networks that are guaranteed not to overfit to any user-specific information. Moreover, summation is perhaps the most primitive functionality in database systems in general, and in private implementations in particular (see, e.g., \cite{kotsogiannis2019privatesql,wilson2019differentially,SYKM17}).

A notable special case is \emph{binary} summation (aka. \emph{counting query}) where each user holds a bit as an input and the goal of the analyzer is to estimate the number of users whose input equals $1$. The vector version of this problem captures, e.g., the case where gradients have been quantized to bits in order to reduce the communication cost  (e.g., the $1$-bit SGD of~\citet{seide20141}). As observed in~\citet{BlumDMN05}, binary summation is of particular interest in ML since it is sufficient for implementing any learning algorithm based on \emph{statistical queries}~\cite{kearns1998efficient}, which includes most of the known PAC-learning algorithms.

Several recent works have studied private summation in the shuffled model \cite{CheuSUZZ19,BalleBGN19,ghazi2019scalable,balle_merged,ghazi2019private,pure-dp-shuffled}. These results achieve DP with parameters $\varepsilon$, $\delta$ (defined in Section~\ref{sec:preliminaries}). 
\citet{CheuSUZZ19} showed that the standard Randomized Response (which goes back to~\citet{warner1965randomized} in the local DP case) is $(\epsilon, \delta)$-DP and incurs a squared error of $O(\frac{1}{\epsilon^2} \cdot \log \frac{1}{\delta})$ with high probability. All mentioned works have also studied \emph{real} summation, culminating in an $(\epsilon,\delta)$-DP protocol in the shuffled model with error arbitrarily close to a discrete Laplace random variable with parameter $1/\epsilon$, and where each user sends $O(1 + \frac{\log(1/\delta)}{\log n})$ messages of $O(\log{n})$ bits each \cite{ghazi2019private, balle_merged}.

{\bf Histograms.} A generalization of the binary summation problem is that of computing \emph{histograms} (aka.~\emph{frequency oracles} or \emph{frequency estimation}), where each user holds an element from some finite set $[B] := \{1,\dots,B\}$ and the goal of the analyzer is to estimate for all $j \in [B]$, the number of users holding element $j$ as input. Computing histograms is fundamental in data analytics and is well-studied in DP (e.g.,~\cite{KBR16,AS19,Suresh19}), as private histogram procedures can be used as a black-box to solve important algorithmic problems such as \emph{heavy hitters} (e.g., \cite{bassily2017practical}) as well as unsupervised machine learning tasks such as \emph{clustering} (e.g., \cite{stemmer2020locally}); furthermore, computing histogram is intimately related to \emph{distribution estimation} (e.g.,~\cite{KBR16,AS19}). In central DP, the smallest possible estimation error for histograms is known to be $\Theta(\min(\frac{\log(1/\delta)}{\epsilon}, \frac{\log {B}}{\epsilon}, n))$ (e.g., Section 7.1 in~\citet{Vadhan-tutorial}). On the other hand, the smallest possible error in local DP is $\Theta(\min(\frac{\sqrt{n \log{B}}}{\epsilon}, n))$ provided $\delta < 1/n$~\cite{bassily2015local}. In the shuffled DP setting and for $\epsilon$ a constant and $\delta$ inverse-polynomial in $n$, the tight estimation error for single-message protocols (where each user sends a single message) is $\tilde \Theta( \min \{ n^{1/4}, \sqrt{B} \})$, whereas multi-message protocols with both error and per-user communication that are logarithmic in $B$ and $n$ are known~\cite{anon-power, erlingsson2020encode}. Recently,~\citet{balcer2019separating} obtained a protocol with error independent of $B$ but logarithmic in $n$, albeit with a per-user communication of $O(B)$ messages each consisting of $O(\log{B})$ bits.

Two recent works \cite{wang2019practical, erlingsson2020encode} studied private histograms in extensions of the shuffled model to multiple shufflers. \citet{wang2019practical} uses Randomized Response whereas~\citet{erlingsson2020encode} uses a fragmented version of RAPPOR \cite{erlingsson2014rappor}.

In this work we focus on the regime where $B  \ll n$, which captures numerous practical scenarios since the number of buckets is typically small compared to the population size.

{\bf Subsequent Works.} Since a conference version of this work~\cite{conference-version} was published, there have been numerous works exploring the shuffled model. Most relevant to us are our follow-up work~\cite{icml21-real-sum} which extends the techniques in this paper to work with real-value summation, and the work of Cheu and Zhilyaev~\cite{Cheu-Zhilyaev-histogramfake} that gave a different protocol for histogram in the shuffled model where each user sends two messages (assuming $n \gg \log B$). Several works have also extended the study of shuffled model beyond simple aggregation tasks; for example,~\cite{BalcerCJM21,ChenG0M21} considers the count distinct problem whereas~\cite{one-round-clustering} considers clustering problems.

\subsection{Main Results}
For the binary summation problem, we give the first  private protocol in the shuffled model achieving mean squared error (MSE) arbitrarily close to the central performance of the Discrete Laplace mechanism while having an \emph{expected} communication per user of $1+o(1)$ messages of 1 bit each.
\begin{theorem}[\bf Binary Summation Protocol]\label{th:bin_agg_nearly_one_bit}
For every $\epsilon \le O(1)$ and every $\delta, \gamma \in (0, 1/2)$, there is an $(\epsilon, \delta)$-DP protocol for binary summation in the multi-message shuffled model with error equal to a Discrete Laplace random variable with parameter $(1-\gamma)\epsilon$ and with an expected communication per user of $1+O\left(\frac{\log^2(1/\delta)}{\gamma \epsilon^2 n}\right)$ bits.
\end{theorem}



We extend Theorem~\ref{th:bin_agg_nearly_one_bit} to a protocol for histograms that, with a moderate number of buckets, has error arbitrarily close to the central DP performance of the Discrete Laplace mechanism while using essentially minimal communication.
\begin{corollary}[\bf Histogram Protocol]\label{th:histograms_single_bit}
For every $\epsilon \le O(1)$ and every $\delta, \gamma \in (0, 1/2)$, there is an $(\epsilon, \delta)$-DP protocol for histograms on sets of size $B$ in the multi-message shuffled model, with error equal to a vector of independent Discrete Laplace random variables each with parameter $\frac{(1-\gamma)\epsilon}{2}$ and with an expected number of messages sent per user equal to $1+O\left(\frac{B\log^2(1/\delta)}{\gamma \epsilon^2 n}\right)$, each consisting of $\lceil\log{B}\rceil + 1$ bits.
\end{corollary}


For the standard setting of constant $\epsilon$ and $\delta$ inverse-polynomial in $n$ and for an arbitrarily small positive constant~$\gamma$, the expected communication per user in Theorem~\ref{th:bin_agg_nearly_one_bit} is $1+o(1)$ bits. Note that $1$ bit of communication per user is required for accurate estimation of the binary summation even in the absence of any privacy constraints.  Likewise, the expected communication per user in Corollary~\ref{th:histograms_single_bit} is $\lceil\log{B}\rceil + 1 + o(1)$ bits.  Here again, $\log{B}$ bits of communication per user is required for accurate estimation of the histogram even in the absence of any privacy constraints.

A natural question in the context of Theorem~\ref{th:bin_agg_nearly_one_bit} and Corollary~\ref{th:histograms_single_bit} is whether the same accuracy and communication can be achieved by a \emph{single-message} protocol in the shuffled model. For histograms, this is impossible given the $\tilde \Omega( \min \{ n^{1/4}, \sqrt{B} \})$ lower bound of~\citet{anon-power} on the $\ell_{\infty}$-error of any single-message protocol whereas the expected $\ell_{\infty}$ error in Corollary~\ref{th:histograms_single_bit} is at most $O(\frac{\log{B}}{\epsilon})$. We prove that this is also impossible for binary summation:
\begin{theorem}[\bf Binary Summation Lower Bound]\label{th:bin_sum_lb_intro}
Let $\delta = 1/n^{\Omega(1)}$ and $\epsilon \le O(1)$. Then, any $(\epsilon, \delta)$-DP protocol for Binary Summation in the single-message shuffled model should incur an expected squared error of at least $\Omega(\log{n})$.
\end{theorem}

In light of the lower bound in Theorem~\ref{th:bin_sum_lb_intro} and the aforementioned lower bound of~\citet{anon-power}, it is striking that the protocols in Theorem~\ref{th:bin_agg_nearly_one_bit} and Corollary~\ref{th:histograms_single_bit} can get arbitrarily close to the central performance of the Discrete Laplace mechanism while being \emph{almost} single-message: the vast majority of users send a single message (consisting of a single bit in the binary summation protocol and $\lceil\log{B}\rceil+1$ bits in the histogram protocol) while only a random $o(1)$ fraction of users sends more than one message!

We point out that, as in previous work in the shuffled and local models of DP, the communication costs in Theorem~\ref{th:bin_agg_nearly_one_bit} and Corollary~\ref{th:histograms_single_bit} exclude the encryption costs. However, as different messages sent by the users have to be encrypted separately, the encryption overhead increases with the number of messages, and hence our almost single-message protocols would be even more appealing compared to other multi-message procedures as in \citet{ghazi2019private, balle_merged,anon-power} when the encryption costs are taken into account.

{\bf Experimental Evaluation.}
We implement our algorithms and compare their performance to several alternatives proposed in the literature, both for binary summation and histogram. For the latter, we evaluate the algorithms on public $1940$ US Census IPUMS dataset, considering both categorical and numerical features. Our experiments support our theoretical analysis: for a broad setting of $n, \eps, \delta$, and $B$, we incur small communication overhead while achieving near-central errors that are noticeably smaller than previous protocols. 
%
The experimental results are presented in the
appendix.

\begin{remark}
We note that our algorithms in Theorem~\ref{th:bin_agg_nearly_one_bit} and Corollary~\ref{th:histograms_single_bit} can be used to learn the empirical distribution of the users' data up a small error. In light of the near-optimality properties of the Discrete Laplace distribution in the central DP model \cite{ghosh2012universally}, our algorithms are also close to optimal. This holds for general error measures including the $\ell_1$, $\ell_2$, and $\ell_{\infty}$ norms, which are well-studied in the literature on distribution estimation and learning (e.g., \cite{KBR16,AS19}).
\end{remark}

\subsection{Overview of Techniques}\label{sec:overview_tech}

Before outlining the proof of Theorem~\ref{th:bin_agg_nearly_one_bit}, we first note that the Discrete Laplace mechanism in the central model  incurs only a \emph{constant} MSE.  To get a similar bound in the shuffled model, any single-message protocol---in particular, Randomized Response and RAPPOR~\cite{erlingsson2014rappor} (which for binary summation coincides with Randomized Response)---is ruled out by Theorem~\ref{th:bin_sum_lb_intro}.  Furthermore, the recent histogram protocols of~\citet{anon-power,erlingsson2020encode}, are also not applicable since they all incur an MSE of $\Omega(\log{n})$.  Theorem~\ref{th:bin_agg_nearly_one_bit} also guarantees vanishing communication overhead: this rules out the split-and-mix protocol in~\citet{ghazi2019private, balle_merged} and a recent protocol of~\citet{pure-dp-shuffled}.

We next recall the prototypical private binary summation procedures in the central setup. If user $i$'s input is $x_i$, then the analyzer simply computes the correct sum $\sum_{i \in [n]} x_i$ and then adds to it a random variable sampled from some probability distribution $\cD$. A common choice of $\cD$ is the Discrete Laplace distribution with parameter $\epsilon$, which yields an $(\epsilon,0)$-DP protocol for binary summation with an asymptotically tight MSE of $O(1/\epsilon^2)$; this is known to be optimal in the central DP model \cite{ghosh2012universally}.

{\bf Using Infinitely Divisible Distributions.}
In order to emulate the prototypical central model mechanism in the shuffled model, we need to distribute both the signal and the noise over the $n$ users. Distributing the signal can be naturally done by having each user $i$ merely send their true input bit $x_i$. Distributing the noise is significantly more challenging since the shuffled model is symmetric and does not allow coordination of noise across users. Consider the framework, captured in Algorithms~\ref{alg:randomizer} and~\ref{alg:analyzer} on page~\pageref{alg:randomizer}, where (a) each user sends (possibly several) bits to the shuffler and (b) the analyzer counts the number of 1s received from the shuffler and outputs it as a proxy for the true sum (possibly after subtracting a fixed bias term). In this case, we would need to decompose the noise random variable into $n$ i.i.d.~non-negative components, and have each user sample and transmit one component in unary. Distributions that are decomposable into the sum of $n$ i.i.d. (not necessarily non-negative) samples for any positive integer $n$ are well-studied in probability theory and are known as \emph{infinitely divisible}. In DP, the Discrete Laplace distribution was observed to be infinitely divisible by~\citet{goryczka2015comprehensive}, and this property  was used by~\citet{balle_merged} for real summation in the shuffled model, albeit with several messages per user, each consisting of $\Omega(\log{n})$ bits. However, decomposing the Discrete Laplace distribution into a sum of i.i.d.~\emph{non-negative} samples---as required by our template above---is clearly impossible since its support contains negative values.

One basic discrete non-negative infinitely divisible distribution is the \emph{Poisson} distribution with parameter $\lambda$, which can be sampled by summing $n$ i.i.d.~samples from a Poisson distribution with parameter $\lambda/n$, for any positive integer $n$. The resulting \emph{Poisson mechanism} can thus be used as a candidate binary summation procedure in the shuffled model. It turns out that this mechanism is $(\epsilon, \delta)$-DP if we set $\lambda$ to $O\left(\frac{\log(1/\delta)}{\epsilon^2}\right)$ (see Theorem~\ref{thm:poisson-dp}). In this case, the expected communication cost of transmitting the per-user noise is equal to the expectation $\lambda/n$, which is much smaller than~$1$. 
We note that, for a reason explained in Section~\ref{subsec:nb}, we can further reduce the communication by considering the Negative Binomial distribution $\NB(r, p)$. This distribution is infinitely divisible as a random sample from $\NB(r, p)$ can be generated by summing $n$ i.i.d.~samples from $\NB(r/n, p)$, for any positive integer $n$.

Unfortunately, it turns out we cannot hope to achieve near-central accuracy using any \emph{non-negative} infinitely divisible distribution. Specifically, we prove in Section~\ref{sec:lb-simple-distributed} that for every such noise distribution, the incurred MSE will grow asymptotically with $\log(1/\delta)$. This is in sharp contrast with the error in the central model, which is independent of $\delta$.

{\bf Unary Encoding and Correlated Noise.}
Instead, the support of our noise distribution has to also contain negative values. To allow this, a natural extension of the above template algorithm is to let each message consists of either an increment (e.g., $+1$) value or a decrement (e.g., $-1$) value. This template leads to a distributed noise strategy that can achieve near-central accuracy, described next. We know from~\citet{goryczka2015comprehensive} that the Discrete Laplace distribution with parameter $\epsilon$ is the same as the distribution of the difference of two independent $\NB(1, e^{-\epsilon})$ random variables, and is thus infinitely divisible. This noise can be distributed in the shuffled model by letting each user sample two independent random variables $Z^1$ and $Z^2$ from $\NB(1/n, e^{-\epsilon})$, and send $Z^1$ increment messages and $Z^2$ decrement messages to the shuffler. This mechanism would achieve the same error as the central Discrete Laplace mechanism. However, since the analyzer can still see the number of increment messages, this scheme is no more private than the (non-negative) mechanism with noise distribution $\NB(1, e^{-\epsilon})$, and thus cannot be $(\epsilon, \delta)$-DP by virtue of the lower bound (Section~\ref{sec:lb-simple-distributed}).

To leverage the power of sending both positive and negative messages, we correlate the input-dependent and noise components sent by the users so that the analyzer is unable to extract much information about the user inputs from one type of messages. We do so by employing a \emph{unary version} of the split-and-mix procedure of~\citet{ishai2006cryptography,ghazi2019private, balle_merged}. Namely, in addition to the aforementioned random variables $Z^1$ and $Z^2$, each user will independently sample a third random variable $Z^3$ from another infinitely divisible distribution, and will send $Z^1 + Z^3$ increment messages and $Z^2 + Z^3$ decrement messages (see Algorithm~\ref{alg:randomizer2} on page~\pageref{alg:randomizer2}). Note that in this case, when $Z^3$ is sufficiently ``spread out'', the analyzer cannot extract much information from counting the number of increments alone, since the noise from $Z^3$ already overwhelms the user inputs. 
We formalize this intuition by proving that, for carefully selected infinitely divisible noise distributions, the resulting mechanism is $(\epsilon, \delta)$-DP and incurs an error that can be made arbitrarily close to that of the central Discrete Laplace mechanism, while incurring an expected communication overhead per user that goes to $0$ with as $n$ increases.


To prove Corollary~\ref{th:histograms_single_bit}, we run the binary summation protocol in parallel on all $B$ buckets, and instead of sending $\pm 1$ valued messages, we concatenate each with the length $\lceil\log{B}\rceil$ binary expansion of the index of the bucket being incremented/decremented.
While a straightforward implementation of the randomizer has a running time of $\Omega(B)$, we show, using the characterization of infinitely divisible distributions in terms of Discrete Compound Poisson (DCP) distributions, that the expected running time can be significantly reduced to the order of the expected per-user communication cost.

{\bf Size-Freeness.}
Infinitely divisible noise mechanisms are much easier to deploy in practice compared to general schemes. This is because for fixed $(\epsilon, \delta)$, the ``noise parameter'' of infinitely divisible mechanisms is independent of the number of users $n$: e.g., in the case of the Poisson Mechanism, the parameter $\lambda$ only depends on $\eps$ and $\delta$. In contrast, computing near-optimal parameters in the shuffled model of the noise parameters for non-infinitely divisible schemes such as Randomized Response~\cite{warner1965randomized} and RAPPOR~\cite{erlingsson2014rappor} requires re-running a time-expensive algorithm for each new value of $n$ (see the appendix for more details). This can be undesirable in practice, especially for real-time applications.


\subsection{Organization}
We provide some background and notation in Section~\ref{sec:preliminaries}. 
In Sections~\ref{sec:dist_mech} and~\ref{sec:corr_dist_mech} we prove Theorem~\ref{th:bin_agg_nearly_one_bit}. 
Our experimental setup and results are presented in Section~\ref{sec:exp_eval_results}. 
We conclude with some open questions in Section~\ref{sec:conc_oqs}.
Corollary~\ref{th:histograms_single_bit} and Theorem~\ref{th:bin_sum_lb_intro} are proved
in Appendices~\ref{sec:red_bit_sum_to_hist} and~\ref{sec:lb_bit_sum_single_message}.

\section{Preliminaries}\label{sec:preliminaries}

Let $n$ be the number of users and $[n] := \{1,\dots,n\}$. In this work, we only consider discrete probability distributions that are supported on (possibly negative) integers.   We write $X \sim \cD$ to denote a random variable $X$ sampled according to $\cD$.  We let $\supp(\cD)$ denote the support of $\cD$, $\E[\cD]$ its mean, and $\Var(\cD)$ its variance. For two distributions $\cD$ and $\cD'$, we denote by $\cD + \cD'$ the distribution of $X + X'$ where $X$ and $X'$ are independently sampled from $\cD$ and $\cD'$ respectively; $\cD - \cD'$ is defined similarly. For integer $k$, we denote by $k + \cD$ the distribution of $k + X$ when $X \sim \cD$. For any $x \in \Z$ we let $\cD(x)$ denote $\Pr_{Z \sim \cD}[Z = x]$. 
We denote by $\bx \sim \bx'$ two datasets (i.e., $n$-dimensional vectors) $\bx$ and $\bx'$ differing on a single user's data (i.e., a single coordinate). 
\begin{definition}[Differential Privacy \cite{dwork2006our,dwork2006calibrating}]\label{def:dp}
For any parameters $\epsilon \geq 0$ and $\delta \in [0,1]$, a randomized mechanism $\cM$ is \emph{$(\epsilon, \delta)$-differentially private (DP)} if for every pair of $\bx \sim \bx'$ and for every subset $\MS$ of transcripts of $\cM$, it holds that $\Pr[\cM(\bx) \in \MS] \leq e^\epsilon \cdot \Pr[\cM(\bx') \in \MS] + \delta$, where the probabilities are over the randomness in $\cM$.
\end{definition}


{\bf Shuffled Model.}
A protocol in the shuffled privacy model consists of three procedures: a \emph{local randomizer} that sends one or several messages depending on its input, a \emph{shuffler} that randomly permutes all incoming messages, and an \emph{analyzer} that takes in the output of the shuffler and returns the final output of the protocol. Privacy is required to be guaranteed with respect to the output of the shuffler.

\section{The $\cD$-Distributed Mechanisms}\label{sec:dist_mech}

In this section, we propose and study a family of simple mechanisms in the shuffled model for the binary summation problem. While the mechanisms in this section do not achieve the accuracy promised in Theorem~\ref{th:bin_agg_nearly_one_bit}, they will serve as an important building block to our eventual algorithm in Section~\ref{sec:corr_dist_mech}. In fact, we will need the privacy guarantee of these mechanisms against a generalization of bit summation called $\Delta$-summation defined as follows. For $\Delta \in \N$, in the \emph{$\Delta$-summation} task, the input to each user is a number $x_i$ in $\{0, \dots, \Delta\}$ and the goal is to compute $\sum_{i \in [n]} x_i$. When $\Delta = 1$, this task is the same as binary summation.

To define our protocol, we first recall that a standard strategy for achieving DP in the central model is to simply add noise to the correct answer. We will refer to such a mechanism the \emph{$\cD$ Mechanism} when the noise distribution is $\cD$.

\begin{definition}
For any distribution $\cD$, the \emph{$\cD$ Mechanism}  for computing a function $f: \cX^n \to \Z^d$ is defined as the mechanism that, on input $\bx \in \cX^n$, outputs $f(\bx) + (Y_1, \dots, Y_d)$ where $Y_i \sim \cD, i \in [d]$ are independent.
\end{definition}

The definition of the $\cD$ Mechanism applies even when $\supp(\cD)$ contains negative integers, but in this section we focus on distributions $\cD$ that are supported on non-negative integers and that are infinitely divisible as defined next.

\begin{definition}[Infinite Divisibility]
A distribution $\cD$ is said to be \emph{infinitely divisible} (abbreviated \emph{$\infty$-div}) if for every $n \in \N$, there exists a distribution $\cD_{/n}$ such that $(X_1 + \cdots + X_n) \sim \cD$ where $X_i \sim \cD_{/n}, i \in [d]$ are independent.
\end{definition}

For an $\infty$-div $\cD$ on non-negative integers, we define the \emph{$\cD$-Distributed Mechanism} in the shuffled model as follows:

\begin{algorithm}[t]
\caption{$\cD$-Distributed Randomizer.} \label{alg:randomizer}
\begin{algorithmic}[1]
\Procedure{Randomizer$_{\cD, n}(x)$}{}
\State Sample $Z \sim \cD_{/n}$ 
\State Send $x + Z$ messages, where each message is 1
\end{algorithmic}
\end{algorithm}
\begin{algorithm}[H]
\caption{$\cD$-Distributed Analyzer.} \label{alg:analyzer}
\begin{algorithmic}[1]
\Procedure{Analyzer$_{\cD}$}{}
\State $U \leftarrow$ number of messages received
\State \Return $U - \E[\cD]$
\EndProcedure
\end{algorithmic}
\end{algorithm}

{\bf Privacy.} Observe that, from the analyzer's perspective, it only sees $U$ (because all the messages are identical) and $U$ is distributed exactly as $\sum_{i \in [n]} x_i + \cD$ by $\infty$-div of $\cD$. From this, we immediately get that the privacy guarantee of the $\cD$-Distributed Mechanism in the \emph{shuffled} model is the same as that of the $\cD$ Mechanism in the \emph{central} model.

\begin{observation} \label{obs:dp-simple}
For any $\epsilon > 0$ and $\delta \in (0,1)$, the $\cD$-Distributed Mechanism is $(\eps, \delta)$-DP in the shuffled model for $\Delta$-summation if and only if the $\cD$ Mechanism is $(\eps, \delta)$-DP in the central model for $\Delta$-summation.
\end{observation}

{\bf Accuracy.} As discussed above, we are guaranteed in Algorithm~\ref{alg:analyzer} that $U = \sum_{i \in [n]} x_i + \cD$, which gives:

\begin{observation} \label{obs:util-simple}
The error of the $\cD$-distributed Mechanism is distributed as $\cD - \E[\cD]$.
\end{observation}

{\bf Expected Communication.} The expected number of messages sent by each user in Algorithm~\ref{alg:randomizer} is $x + \E[\cD_{/n}] = x + \frac{\E[\cD]}{n}$. Using the fact that $x \in \{0, \dots, \Delta\}$, we get:

\begin{observation} \label{obs:com-simple}
The expected number of messages sent by a user in the $\cD$-Distributed Mechanism is at most $\Delta + \frac{\E[\cD]}{n}$.
\end{observation}

\subsection{Example I: The Poisson Mechanism}

Arguably, the simplest protocol in the family of $\cD$-Distributed Mechanisms is the \emph{Poisson Mechanism} that uses the Poisson distribution%
\footnote{$\Poi(\lambda)$ is defined as $\Poi(k; \lambda) = \lambda^k e^{-k} / k!$.}, 
which is $\infty$-div. In this protocol, $\cD$ is $\Poi(\lambda)$ for some $\lambda \in \R^+$ and $\cD_{/n}$ is simply $\Poi(\lambda/n)$.  We now compute its privacy guarantee.

\begin{theorem} \label{thm:poisson-dp}
For any $\epsilon > 0$, $\delta \in (0,1)$, and $\Delta \in \N$, the $\Poi(\lambda)$ Mechanism with
$\lambda = \frac{16 \log(10/\delta)}{(1 - e^{-\eps/\Delta})^2} + \frac{2\Delta}{1 - e^{-\eps/\Delta}}$,
is $(\eps, \delta)$-DP in the central model for $\Delta$-summation.
\end{theorem}

By setting $\Delta = 1$ and using Observations~\ref{obs:dp-simple},~\ref{obs:util-simple}, and~\ref{obs:com-simple}, we get the following for binary summation. 

\begin{corollary} \label{lem:poisson-shuffled-dp}
For any $0 < \eps \le O(1)$ and $\delta \in (0,1)$, let $\lambda$ be as in Theorem~\ref{thm:poisson-dp} with $\Delta = 1$. The $\Poi(\lambda)$-Distributed Mechanism is $(\eps, \delta)$-DP for binary summation in the shuffled model, each user sends at most $1 + O\left(\frac{\log(1/\delta)}{\eps^2 n}\right)$ one-bit messages in expectation, and the MSE is $O\left(\frac{\log(1/\delta)}{\eps^2}\right)$.
\end{corollary}

\subsection{Example II: The Negative Binomial Mechanism}
\label{subsec:nb}

A disadvantage of the Poisson Mechanism is that, in the most important regime where $\frac{\eps}{\Delta} \ll 1$, the expected number of messages sent is $1 + O\left(\frac{\log(1/\delta)}{n} \cdot \left(\frac{\Delta}{\eps}\right)^2\right)$. In this subsection, we show how to reduce the dependency on $\frac{\Delta}{\eps}$ from $(\frac{\Delta}{\eps})^2$ to $\frac{\Delta}{\eps}$, while retaining a similar error bound. (As we will see in Section~\ref{sec:corr_dist_mech}, this dependency will also permeate to our eventual algorithm in the proof of Theorem~\ref{th:bin_agg_nearly_one_bit}.) Before doing so, we note that the $(\frac{\Delta}{\eps})^2$ dependency is necessary for the Poisson Mechanism since it is well-known%
\footnote{This follows from the sensitivity of $\Delta$-summation; see, e.g..~\citet{Vadhan-tutorial}.}
that the MSE of any central $(\eps, o(1))$-DP protocol for $\Delta$-summation has to be at least $\Omega((\frac{\Delta}{\eps})^2)$ and since the MSE of the $\Poi(\lambda)$ Mechanism is exactly $\lambda$, we must hence set $\lambda \geq (\frac{\Delta}{\eps})^2$. Thus, the expected number of messages sent per user in the $\Poi(\lambda)$ Mechanism must be $1 + \frac{\lambda}{n} \geq 1 + \Omega\left(\frac{1}{n} \cdot \left(\frac{\Delta}{\eps}\right)^2\right)$.

To circumvent this, we first observe that the above argument holds only because the parameter $\lambda$ of the Poisson Mechanism governs both the MSE (i.e., the variance of the distribution) and the number of messages sent (i.e., the expectation of the distribution). This motivates us to seek an $\infty$-div distribution on the negative integers whose variance and mean can be very different. We consider the negative binomial distribution\footnote{$\NB(r, p)$ is defined as $\NB(k; r, p) = {k+r-1 \choose k} (1-p)^r p^k$.  We remark that the negative binomial distribution may be viewed as a generalization of the Poisson distribution: when taking $r \to \infty$ and letting $p = \lambda/r$, $\NB(r, p)$ point-wise converges to $\Poi(\lambda)$.} $\NB(r, p)$ which is $\infty$-div as $\NB(r, p) = \sum_{i=1}^n \NB(\frac{r}{n}, p)$ for every $n \in \N$. Moreover, $\E[\NB(r, p)] = \frac{pr}{(1 - p)}$ and $\Var[\NB(r, p)] = \frac{pr}{(1 - p)^2}$, which can be very different when $p$ is close to $1$. We next show a DP guarantee for this Negative Binomial Mechanism\footnote{In the conference version of this work~\cite{conference-version}, the values of $r, p$ set in Theorem~\ref{thm:nb-privacy} contained an error. This has now been fixed in this version. Note that this does not affect the experiment results, since we use numerical approaches to compute privacy parameters in the experiments anyway (see the appendix).}.

\begin{theorem}\label{thm:nb-privacy}
For any $\eps, \delta \in (0, 1)$ and $\Delta \in \N$, let $p = e^{-0.2\eps/\Delta}$ and $r = 3\left(1  + \log\left(\frac{1}{\delta}\right)\right)$. The $\NB(r, p)$ Mechanism is $(\eps, \delta)$-DP in the central model for $\Delta$-summation.
\end{theorem}

Plugging $\Delta = 1$, we get the following corollary. When compared to Poisson Mechanism (Corollary~\ref{lem:poisson-shuffled-dp}), we achieve the same error bound but with a $\frac{1}{\eps}$ instead of $\frac{1}{\eps^2}$ multiplicative term in the expected number of additional messages sent.

\begin{corollary} \label{lem:nb-shuffled-dp}
For any $\eps, \delta \in (0, 1)$ with $\eps \leq O(1)$, let $p, r$ be as in Theorem~\ref{thm:nb-privacy} with $\Delta = 1$. The $\NB(r, p)$-Distributed Mechanism is $(\eps, \delta)$-DP for binary summation in the shuffled model, each user sends at most $1 + O\left(\frac{\log(1/\delta)}{\eps n}\right)$ one-bit messages in expectation, and the MSE is $O\left(\frac{\log(1/\delta)}{\eps^2}\right)$.
\end{corollary}

\subsection{A Lower Bound for $\cD$-Distributed Mechanisms}
\label{sec:lb-simple-distributed}

The downside of the Poisson and Negative Binomial Mechanisms is that they suffer from an MSE of $O_{\eps}(\log(\frac{1}{\delta}))$ instead of the $O(\frac{1}{\eps^2})$ MSE, independent of $\delta$, of the central Discrete Laplace Mechanism. It turns out that this dependency on $\log(\frac{1}{\delta})$ is necessary for every $\cD$-Distributed Mechanism:

\begin{lemma} \label{lem:lb-distributed-mechanism}
For any infinitely divisible distribution $\cD$ on non-negative integers, if $\cD$-Distributed Mechanism is $(\eps, \delta)$-DP in the shuffled model for binary summation, then the MSE of the mechanism is $\Omega_{\eps}(\log(1/\delta))$.
\end{lemma}

In other words, $\cD$-Distributed Mechanisms do not suffice for the goal of achieving near-central error guarantees. 

\section{Correlated Distributed Mechanisms}\label{sec:corr_dist_mech}

We next present a family of protocols in the shuffled model that will overcome the lower bound barrier of Lemma~\ref{lem:lb-distributed-mechanism} and achieve an accuracy/privacy trade-off arbitrarily close to the central model. We start by outlining the intuition behind the protocol. First, suppose hypothetically that we could somehow implement the $\cD$ Mechanism in the shuffled model when $\supp(\cD)$ can contain negative integers. Then, we would actually be done! This is because the Discrete Laplace distribution%
\footnote{
$\DLap(s)$ is defined as 
$\DLap(k; s) = \frac{1}{C(s)} \cdot e^{- |k| \cdot s}$, where $C(s) = \sum_{k = -\infty}^{\infty} e^{-|k| \cdot s}$ is the normalization constant.
} is $\infty$-div as $\DLap(\eps) = \NB(1, e^{-\eps}) - \NB(1, e^{-\eps})$ (see, e.g., \cite{kotz2001book}), and $\NB(1, e^{-\eps})$ is $\infty$-div, and is in fact the geometric distribution. Of course, the problem is that if we sample the number of messages from $\DLap(\eps)_{/n}$ we will often try to send a \emph{negative} number of messages, which is meaningless! 

This leads us to using two types of messages: one for increments (denoted $+1$) and one for decrements (denoted $-1$). The $+1$ messages are sampled as in the $\NB(1, e^{-\eps})$ Mechanism, and so are the $-1$ messages except that each user pretends that their input is $0$ in this case. The analyzer's answer is the difference between the number of $+1$ messages and the number of $-1$ messages it receives. The aforementioned fact that the difference of two $\NB(1, e^{-\eps})$ random variables is $\DLap(\eps)$ implies that this protocol has the same accuracy as the central Discrete Laplace Mechanism, as desired. However, this protocol is not $(\eps, \delta)$-DP in the shuffled model. To see this, note that the analyzer sees the number of $+1$ messages received, and hence the protocol is no more private than the $\NB(1, e^{-\eps})$ Mechanism, which is not $(\eps, \delta)$-DP as explained by Lemma~\ref{lem:lb-distributed-mechanism}. To overcome this, we ``mask'' the numbers of $+1$ and $-1$ messages by sampling a random variable $Z$ from $\cD_{/n}$ for some other $\infty$-div $\cD$ over non-negative integers, and additionally send $Z$ increments (i.e., $+1$) and $Z$ decrements (i.e., $-1$). Clearly, this does not affect the accuracy of the analyzer, but we will show that it improves privacy.
For every choice of $\infty$-div $\cD^1, \cD^2, \cD^3$ on non-negative integers, we define the \emph{$(\cD^1, \cD^2, \cD^3)$-Correlated Distributed Mechanism}:
\begin{minipage}{0.48\textwidth}
\begin{algorithm}[H]
\caption{\small $(\cD^1, \cD^2, \cD^3)$-Correlated Distributed Randomizer} \label{alg:randomizer2}
\begin{algorithmic}[1]
\Procedure{Randomizer$_{\cD^1, \cD^2, \cD^3, n}(x)$}{}
\State Sample $Z^1 \sim \cD^1_{/n}$
\State Sample $Z^2 \sim \cD^2_{/n}$
\State Sample $Z^3 \sim \cD^3_{/n}$
\State Send $x + Z^1 + Z^3$ many $+1$ messages.
\State Send $Z^2 + Z^3$ many $-1$ messages.
\EndProcedure
\end{algorithmic}
\end{algorithm}
\end{minipage}
\hfill
\begin{minipage}{0.48\textwidth}
\begin{algorithm}[H]
\caption{\small $(\cD^1, \cD^2, \cD^3)$-Correlated Distributed Analyzer} \label{alg:analyzer2}
\begin{algorithmic}[1]
\Procedure{Analyzer$_{\cD^1, \cD^2}$}{}
\State $U_{+1} \leftarrow$ number of $+1$ messages received.
\State $U_{-1} \leftarrow$ number of $-1$ messages received.
\State \Return $U_{+1} - U_{-1} - \E[\cD^1 - \cD^2]$
\EndProcedure
\end{algorithmic}
\end{algorithm}
\end{minipage}

\textbf{Accuracy and Communication Complexity.}
The accuracy and communication complexity of the protocol can be derived as in the $\cD$-Distributed Mechanism from Section~\ref{sec:dist_mech}:

\begin{observation} \label{obs:util-neg}
The error of $(\cD^1, \cD^2, \cD^3)$-Distributed Mechanism is distributed as $(\cD^1 - \cD^2) - \E[\cD^1 - \cD^2]$.
\end{observation}

\begin{observation} \label{obs:com-neg}
The expected number of messages sent by each user in the $(\cD^1, \cD^2, \cD^3)$-Distributed Mechanism is at most $1 + \frac{\E[\cD^1] + \E[\cD^2] + 2 \cdot \E[\cD^3]}{n}$.
\end{observation}

\textbf{Privacy.}
The crux of our DP proof for the $(\cD^1, \cD^2, \cD^3)$-Correlated Distributed Mechanism is the following theorem:

\begin{theorem} \label{thm:dp-neg-noise}
Let $\Delta > 0$ and $\cD^1, \cD^2, \cD^3$ satisfy:
\begin{itemize}[nosep]
\item (Privacy of True Noise) The $(\cD^1 - \cD^2)$ Mechanism is $\eps_1$-DP for binary summation in the central model.
\item (Privacy of Correlated Noise) The $\cD^3$ Mechanism is $(\eps_2, \delta_2)$-DP for $\Delta$-summation in the central model. 
\item (Concentration of Neg.~Noise) $\Pr_{Y \sim \cD^2}[Y > \Delta] \leq \delta_3$.
\end{itemize}
Then, the $(\cD^1, \cD^2, \cD^3)$-Correlated Distributed Mechanism is $(\eps_1 + \eps_2, e^{\eps_1} \cdot \delta_2 + 2 e^{2\eps_1} \cdot \delta_3)$-DP in the shuffled model.
\end{theorem}

\textbf{Proof Overview.}
All the analyzer sees is $(U_{+1}, U_{-1})$ $= (\sum_{i \in [n]} x_i + \hZ^1 + \hZ^3, \hZ^2 + \hZ^3)$ where for $j \in [3]$, 
$\hZ^j \sim \cD^j$. Since there is bijection between $(U_{+1}, U_{-1})$ and $(U_{+1} - U_{-1}, U_{-1}) = (\sum_{i \in [n]} x_i + \hZ^1 - \hZ^2, \hZ^2 + \hZ^3)$, we may consider the distribution on the latter. The first coordinate is the same as the analyzer's view from the $(\cD^1 - \cD^2)$ Mechanism, which is $\eps_1$-DP by the first assumption. Once we condition on $U_{+1} - U_{-1}$ being equal to some value, we are left to consider a distribution on $U_{-1}$. This distribution is not the same as the original one (before conditioning) since $U_{-1}$ and $U_{+1} - U_{-1}$ are correlated. But this correlation only comes via $\hZ^2$, as $\hZ^3$ does not appear in $U_{+1} - U_{-1}$. By the concentration of $\cD^2$ (the third assumption), 
$\hZ^2$ is rarely larger than $\Delta$. When $\hZ^2 \leq \Delta$, we can use the $(\eps_2, \delta_2)$-DP of the $\cD^3$ Mechanism for $\Delta$-summation to argue the privacy of the protocol. The proof follows this intuition while carefully tracking the privacy loss in each step.

\subsection{Near-Central Accuracy with Shuffled Mechanisms}

We use Theorem~\ref{thm:dp-neg-noise} to derive our ``near-central'' protocol (Theorem~\ref{th:bin_agg_nearly_one_bit}).  Specifically, we set $\cD^1, \cD^2$ so that $\cD^1 - \cD^2 = \DLap(0.99\eps)$, which implies that the $(\cD^1 - \cD^2)$ Mechanism is $0.99\eps$-DP in the central model. The remaining privacy budget of $0.01\eps$ is allocated to the $\cD^3$ Mechanism, which we set to be the $\NB(\cdot. \cdot)$ Mechanism with appropriate parameters. We use the Negative Binomial rather than the Poisson distribution as the former (Theorem~\ref{thm:nb-privacy}) has a smaller communication cost than the latter (Theorem~\ref{thm:poisson-dp}). 

\section{Experimental Evaluation and Results}\label{sec:exp_eval_results}

\subsection{Binary Summation}

In this section, we evaluate our protocols. Specifically, we consider the Poisson Distributed Mechanism (Theorem~\ref{thm:poisson-dp}) 
and the Correlated Distributed Noise Mechanism (Theorems~\ref{th:bin_agg_nearly_one_bit} and~\ref{thm:dp-neg-noise}).
We consider the root mean square error (RMSE), which is independent of the input data for all methods considered, and for this reason there is no need to consider performance on particular data sets. For each setting of $\eps, \delta$, our parameters are selected in an accurate manner; this is explained in detail in the
Appendix~\ref{app:bin-param-computation}. 
We only note here that for the Correlated Distributed Mechanism, we set $\cD^1 = \cD^2 = \NB(1, e^{-\eps_1})$ with $\eps_1$ such that the RMSE of the protocol is 20\% more than that of the (central) $\DLap(\eps)$ Mechanism. We compare our algorithms against the classic \emph{Randomized Response (RR)} algorithm, where a user with input $x$ sends $x$ w.p. $p$, and $1 - x$ w.p. $1 - p$, for some parameter $p\in [0,1/2]$.

\textbf{Error.}
We compute the errors as $\eps$ or $\delta$ varies. The corresponding plots are shown in Figure~\ref{fig:err-bin}; we include the error of the (central) Discrete Laplace Mechanism for comparison (but of course this is not directly implementable in the shuffled model). We remark that 
the RMSEs of our protocols are independent of the number of users $n$, and only the RMSE of RR depends on $n$, which we choose to be 10,000. 
While the Correlated Distributed protocol has a constant RMSE as we vary $\delta$, both Poisson and RR incur larger RMSEs. In particular, when $\delta = 10^{-6}, \eps = 1$, the RMSE of the Correlated Distributed protocol is 3.5 times less than that of Poisson and RR (which essentially coincide). The fact that the Poisson Mechanism and RR have essentially the same RMSE should come as no surprise, since the binomial distribution $\Bin(n, p)$ converges (in the distributional sense) to the Poisson distribution $\Poi(np)$ as $n \to \infty$ and $p$ is kept constant. This means that we would prefer RR in this case, since it always sends one message.

\begin{figure}[h]
\centering
\includegraphics[width=0.35\textwidth]{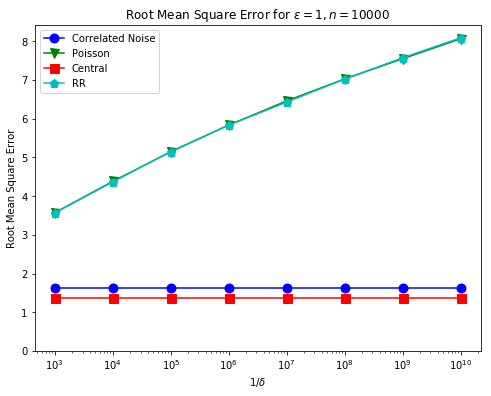}
\includegraphics[width=0.35\textwidth]{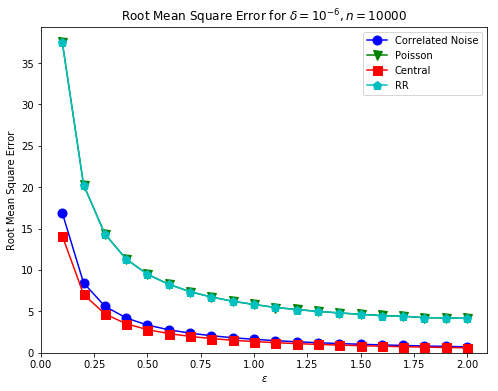}
\caption{RMSE of the protocols for binary summation. 
Note that the RMSEs of RR and Poisson are essentially the same.}
\label{fig:err-bin}
\end{figure}

\textbf{Communication Complexity.}
Figure~\ref{fig:msg-bin} shows the plots of the expected number of additional messages sent by each user in our Poisson and Correlated Distributed Mechanisms. Here we let $n = 10,000$. In the reasonable setting where $\delta = 10^{-6}$ and $\eps = 1$, the expected number of additional messages sent in the Correlated Distributed Mechanism is only $0.04$, whereas in the Poisson Mechanism, it is even smaller at $0.0003$. Even in the more extreme case of $\eps = 0.1$, the former is still $0.278$ and the latter is only $0.141$. 

\begin{figure}[h]
\centering
\includegraphics[width=0.35\textwidth]{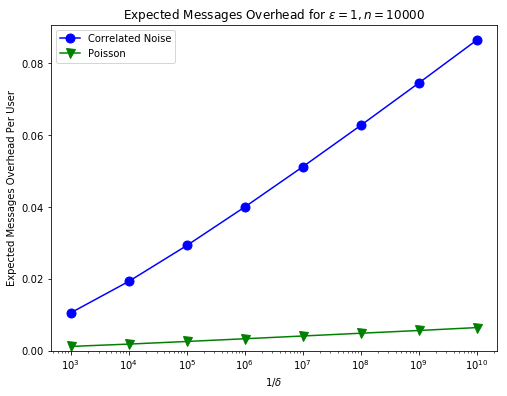}
\includegraphics[width=0.35\textwidth]{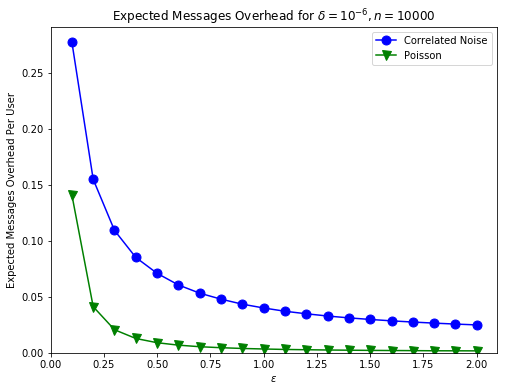}
\caption{Expected additional number of messages sent by each user for binary summation. While we use $n = 10^4$, this expectation scales linearly in $1/n$. E.g., if $n = 10^5$, the plots will look the same, but with the $y$-axis scaled down by a factor of~10.}
\label{fig:msg-bin}
\end{figure}

\subsection{Histograms}

We have also performed experiments on the histogram versions of our Correlated Distributed and Poisson Mechanisms, and compare them against three algorithms from the literature: $B$-Randomized Response ($B$-RR), RAPPOR~\cite{erlingsson2014rappor}, and Fragmented RAPPOR~\cite{erlingsson2020encode}. Each of these three can be viewed as a $B$-ary generalization of the binary RR. We ran these algorithms on two IPUMS datasets~\cite{sobek2010integrated}. Due to space constraints, the full description of the experiments and results are deferred to
Appendix~\ref{app:exp-histogram}.
Here we just summarize our findings:
For most parameters, RAPPOR incurs significantly larger errors than $B$-RR. Moreover, similarly to how RR mirrors the Poisson Mechanism in the binary case, Fragmented RAPPOR gives almost the same results as Poisson for histograms. In terms of RMSE, our correlated mechanism is significantly closer to the central model error than its competitors. The plots are in fact very similar to those of binary summation for the Correlated Distributed and Poisson Mechanisms, with RR doing 25-50\% worse than Poisson. In terms of $\ell_{\infty}$, RR incurs more than $3\times$ $\ell_\infty$ errors compared to our algorithms. Furthermore, as corroborated by theory~\cite{anon-power}, the error of RR grows quickly with the number $B$ of buckets, while those of our mechanisms grow very slowly.

\section{Conclusions and Open Questions}\label{sec:conc_oqs}
We proposed DP algorithms in the shuffled model for binary summation and histograms with accuracy arbitrarily close to the central model and a vanishing communication overhead. There are several questions left open by our work. One is to obtain algorithms achieving near-central performance with negligible communication overhead for the problems of real summation\footnote{This is largely resolved in the follow-up work~\cite{icml21-real-sum} but there are still a few remaining open questions there; please refer to the conclusion of that paper for more details.}, vector summation, and histograms over a large number $B \geq \Omega(n)$ of buckets. Another question is to prove a lower bound against \emph{expected single-message} protocols (that can send $0$, $1$ or more messages in the worst-case) as our lower bound in Theorem~\ref{th:bin_sum_lb_intro} does not hold in this case. A very interesting related direction is to build on our algorithms to improve the communication complexity of DP stochastic gradient descent in a federated learning setup.

\section*{Acknowledgements}
We would like to thank Noah Golowich, Laura Peskin and Alex Zani for insightful discussions and feedback.


\balance
\bibliographystyle{icml2020}
\bibliography{refs}

\begin{thebibliography}{59}
\providecommand{\natexlab}[1]{#1}
\providecommand{\url}[1]{\texttt{#1}}
\expandafter\ifx\csname urlstyle\endcsname\relax
  \providecommand{\doi}[1]{doi: #1}\else
  \providecommand{\doi}{doi: \begingroup \urlstyle{rm}\Url}\fi

\bibitem[Abowd(2018)]{abowd2018us}
Abowd, J.~M.
\newblock The {US Census Bureau} adopts differential privacy.
\newblock In \emph{KDD}, pp.\  2867--2867, 2018.

\bibitem[Acharya \& Sun(2019)Acharya and Sun]{AS19}
Acharya, J. and Sun, Z.
\newblock Communication complexity in locally private distribution estimation
  and heavy hitters.
\newblock In \emph{ICML}, pp.\  51--60, 2019.

\bibitem[Ali \& Silvey(1966)Ali and Silvey]{ali1966general}
Ali, S.~M. and Silvey, S.~D.
\newblock A general class of coefficients of divergence of one distribution
  from another.
\newblock \emph{JRS: Series B (Methodological)}, 28\penalty0 (1):\penalty0
  131--142, 1966.

\bibitem[{Apple Differential Privacy Team}(2017)]{dp2017learning}
{Apple Differential Privacy Team}.
\newblock Learning with privacy at scale.
\newblock \emph{Apple Machine Learning Journal}, 2017.

\bibitem[Balcer \& Cheu(2020)Balcer and Cheu]{balcer2019separating}
Balcer, V. and Cheu, A.
\newblock Separating local \& shuffled differential privacy via histograms.
\newblock In \emph{ITC}, pp.\  1:1--1:14, 2020.

\bibitem[Balcer et~al.(2021)Balcer, Cheu, Joseph, and Mao]{BalcerCJM21}
Balcer, V., Cheu, A., Joseph, M., and Mao, J.
\newblock Connecting robust shuffle privacy and pan-privacy.
\newblock In Marx, D. (ed.), \emph{Proceedings of the 2021 {ACM-SIAM} Symposium
  on Discrete Algorithms, {SODA} 2021, Virtual Conference, January 10 - 13,
  2021}, pp.\  2384--2403. {SIAM}, 2021.
\newblock \doi{10.1137/1.9781611976465.142}.
\newblock URL \url{https://doi.org/10.1137/1.9781611976465.142}.

\bibitem[Balle et~al.(2019)Balle, Bell, Gasc{\'{o}}n, and Nissim]{BalleBGN19}
Balle, B., Bell, J., Gasc{\'{o}}n, A., and Nissim, K.
\newblock The privacy blanket of the shuffle model.
\newblock In \emph{CRYPTO}, pp.\  638--667, 2019.

\bibitem[Balle et~al.(2020)Balle, Bell, Gasc{\'{o}}n, and Nissim]{balle_merged}
Balle, B., Bell, J., Gasc{\'{o}}n, A., and Nissim, K.
\newblock Private summation in the multi-message shuffle model.
\newblock \emph{arXiv: 2002.00817}, 2020.

\bibitem[Barthe \& Olmedo(2013)Barthe and Olmedo]{barthe2013beyond}
Barthe, G. and Olmedo, F.
\newblock Beyond differential privacy: Composition theorems and relational
  logic for $f$-divergences between probabilistic programs.
\newblock In \emph{ICALP}, pp.\  49--60, 2013.

\bibitem[Bassily \& Smith(2015)Bassily and Smith]{bassily2015local}
Bassily, R. and Smith, A.
\newblock Local, private, efficient protocols for succinct histograms.
\newblock In \emph{STOC}, pp.\  127--135, 2015.

\bibitem[Bassily et~al.(2017)Bassily, Nissim, Stemmer, and
  Thakurta]{bassily2017practical}
Bassily, R., Nissim, K., Stemmer, U., and Thakurta, A.~G.
\newblock Practical locally private heavy hitters.
\newblock In \emph{NIPS}, pp.\  2288--2296, 2017.

\bibitem[Beimel et~al.(2008)Beimel, Nissim, and Omri]{beimel2008distributed}
Beimel, A., Nissim, K., and Omri, E.
\newblock Distributed private data analysis: Simultaneously solving how and
  what.
\newblock In \emph{CRYPTO}, pp.\  451--468, 2008.

\bibitem[Bittau et~al.(2017)Bittau, Erlingsson, Maniatis, Mironov, Raghunathan,
  Lie, Rudominer, Kode, Tinn{\'{e}}s, and Seefeld]{bittau17}
Bittau, A., Erlingsson, {\'{U}}., Maniatis, P., Mironov, I., Raghunathan, A.,
  Lie, D., Rudominer, M., Kode, U., Tinn{\'{e}}s, J., and Seefeld, B.
\newblock Prochlo: Strong privacy for analytics in the crowd.
\newblock In \emph{SOSP}, pp.\  441--459, 2017.

\bibitem[Blum et~al.(2005)Blum, Dwork, McSherry, and Nissim]{BlumDMN05}
Blum, A., Dwork, C., McSherry, F., and Nissim, K.
\newblock Practical privacy: the {SuLQ} framework.
\newblock In \emph{PODS}, pp.\  128--138, 2005.

\bibitem[Canonne(2017)]{poisson-tail}
Canonne, C.
\newblock A short note on {P}oisson tail bounds, 2017.
\newblock URL
  \url{http://www.cs.columbia.edu/~ccanonne/files/misc/2017-poissonconcentration.pdf}.

\bibitem[Chan et~al.(2012)Chan, Shi, and Song]{ChanSS12}
Chan, T.~H., Shi, E., and Song, D.
\newblock Optimal lower bound for differentially private multi-party
  aggregation.
\newblock In \emph{ESA}, pp.\  277--288, 2012.

\bibitem[Chang et~al.(2021)Chang, Ghazi, Kumar, and
  Manurangsi]{one-round-clustering}
Chang, A., Ghazi, B., Kumar, R., and Manurangsi, P.
\newblock Locally private k-means in one round.
\newblock \emph{CoRR}, abs/2104.09734, 2021.
\newblock URL \url{https://arxiv.org/abs/2104.09734}.

\bibitem[Chen et~al.(2021)Chen, Ghazi, Kumar, and Manurangsi]{ChenG0M21}
Chen, L., Ghazi, B., Kumar, R., and Manurangsi, P.
\newblock On distributed differential privacy and counting distinct elements.
\newblock In Lee, J.~R. (ed.), \emph{12th Innovations in Theoretical Computer
  Science Conference, {ITCS} 2021, January 6-8, 2021, Virtual Conference},
  volume 185 of \emph{LIPIcs}, pp.\  56:1--56:18. Schloss Dagstuhl -
  Leibniz-Zentrum f{\"{u}}r Informatik, 2021.
\newblock \doi{10.4230/LIPIcs.ITCS.2021.56}.
\newblock URL \url{https://doi.org/10.4230/LIPIcs.ITCS.2021.56}.

\bibitem[Cheu \& Zhilyaev(2021)Cheu and Zhilyaev]{Cheu-Zhilyaev-histogramfake}
Cheu, A. and Zhilyaev, M.
\newblock Differentially private histograms in the shuffle model from fake
  users.
\newblock \emph{CoRR}, abs/2104.02739, 2021.
\newblock URL \url{https://arxiv.org/abs/2104.02739}.

\bibitem[Cheu et~al.(2019)Cheu, Smith, Ullman, Zeber, and Zhilyaev]{CheuSUZZ19}
Cheu, A., Smith, A.~D., Ullman, J., Zeber, D., and Zhilyaev, M.
\newblock Distributed differential privacy via shuffling.
\newblock In \emph{EUROCRYPT}, pp.\  375--403, 2019.

\bibitem[Csisz{\'a}r(1964)]{csiszar1964informationstheoretische}
Csisz{\'a}r, I.
\newblock Eine informationstheoretische ungleichung und ihre anwendung auf
  beweis der ergodizitaet von markoffschen ketten.
\newblock \emph{Magyer Tud. Akad. Mat. Kutato Int. Koezl.}, 8:\penalty0
  85--108, 1964.

\bibitem[Csisz{\'a}r(1967)]{csiszar1967information}
Csisz{\'a}r, I.
\newblock Information-type measures of difference of probability distributions
  and indirect observation.
\newblock \emph{studia scientiarum Mathematicarum Hungarica}, 2:\penalty0
  229--318, 1967.

\bibitem[Ding et~al.(2017)Ding, Kulkarni, and Yekhanin]{ding2017collecting}
Ding, B., Kulkarni, J., and Yekhanin, S.
\newblock Collecting telemetry data privately.
\newblock In \emph{NIPS}, pp.\  3571--3580, 2017.

\bibitem[Dwork et~al.(2006{\natexlab{a}})Dwork, Kenthapadi, McSherry, Mironov,
  and Naor]{dwork2006our}
Dwork, C., Kenthapadi, K., McSherry, F., Mironov, I., and Naor, M.
\newblock Our data, ourselves: Privacy via distributed noise generation.
\newblock In \emph{EUROCRYPT}, pp.\  486--503, 2006{\natexlab{a}}.

\bibitem[Dwork et~al.(2006{\natexlab{b}})Dwork, McSherry, Nissim, and
  Smith]{dwork2006calibrating}
Dwork, C., McSherry, F., Nissim, K., and Smith, A.
\newblock Calibrating noise to sensitivity in private data analysis.
\newblock In \emph{TCC}, pp.\  265--284, 2006{\natexlab{b}}.

\bibitem[Dwork et~al.(2014)Dwork, Roth, et~al.]{dwork2014algorithmic}
Dwork, C., Roth, A., et~al.
\newblock {The Algorithmic Foundations of Differential Privacy}.
\newblock \emph{Foundations and Trends{\textregistered} in Theoretical Computer
  Science}, 9\penalty0 (3--4):\penalty0 211--407, 2014.

\bibitem[Erlingsson et~al.(2014)Erlingsson, Pihur, and
  Korolova]{erlingsson2014rappor}
Erlingsson, {\'U}., Pihur, V., and Korolova, A.
\newblock {RAPPOR}: Randomized aggregatable privacy-preserving ordinal
  response.
\newblock In \emph{CCS}, pp.\  1054--1067, 2014.

\bibitem[Erlingsson et~al.(2019)Erlingsson, Feldman, Mironov, Raghunathan,
  Talwar, and Thakurta]{erlingsson2019amplification}
Erlingsson, {\'U}., Feldman, V., Mironov, I., Raghunathan, A., Talwar, K., and
  Thakurta, A.
\newblock Amplification by shuffling: From local to central differential
  privacy via anonymity.
\newblock In \emph{SODA}, pp.\  2468--2479, 2019.

\bibitem[Erlingsson et~al.(2020)Erlingsson, Feldman, Mironov, Raghunathan,
  Song, Talwar, and Thakurta]{erlingsson2020encode}
Erlingsson, {\'U}., Feldman, V., Mironov, I., Raghunathan, A., Song, S.,
  Talwar, K., and Thakurta, A.
\newblock Encode, shuffle, analyze privacy revisited: Formalizations and
  empirical evaluation.
\newblock \emph{arXiv:2001.03618}, 2020.

\bibitem[Feller(1968)]{Feller}
Feller, W.
\newblock \emph{An Introduction to Probability Theory and Its Applications},
  volume~1.
\newblock Wiley, January 1968.

\bibitem[Ghazi et~al.(2019{\natexlab{a}})Ghazi, Golowich, Kumar, Pagh, and
  Velingker]{anon-power}
Ghazi, B., Golowich, N., Kumar, R., Pagh, R., and Velingker, A.
\newblock On the power of multiple anonymous messages.
\newblock Cryptology ePrint Archive:1382, 2019{\natexlab{a}}.

\bibitem[Ghazi et~al.(2019{\natexlab{b}})Ghazi, Pagh, and
  Velingker]{ghazi2019scalable}
Ghazi, B., Pagh, R., and Velingker, A.
\newblock Scalable and differentially private distributed aggregation in the
  shuffled model.
\newblock \emph{arXiv: 1906.08320}, 2019{\natexlab{b}}.

\bibitem[Ghazi et~al.(2020{\natexlab{a}})Ghazi, Kumar, Manurangsi, and
  Pagh]{conference-version}
Ghazi, B., Kumar, R., Manurangsi, P., and Pagh, R.
\newblock Private counting from anonymous messages: Near-optimal accuracy with
  vanishing communication overhead.
\newblock In \emph{ICML}, pp.\  3505--3514, 2020{\natexlab{a}}.

\bibitem[Ghazi et~al.(2020{\natexlab{b}})Ghazi, Kumar, Manurangsi, Pagh, and
  Velingker]{pure-dp-shuffled}
Ghazi, B., Kumar, R., Manurangsi, P., Pagh, R., and Velingker, A.
\newblock Pure differentially private summation from anonymous messages.
\newblock In \emph{ITC}, pp.\  15:1--15:23, 2020{\natexlab{b}}.

\bibitem[Ghazi et~al.(2020{\natexlab{c}})Ghazi, Manurangsi, Pagh, and
  Velingker]{ghazi2019private}
Ghazi, B., Manurangsi, P., Pagh, R., and Velingker, A.
\newblock Private aggregation from fewer anonymous messages.
\newblock In \emph{EUROCRYPT}, pp.\  798--827, 2020{\natexlab{c}}.

\bibitem[Ghazi et~al.(2021)Ghazi, Kumar, Manurangsi, Pagh, and
  Sinha]{icml21-real-sum}
Ghazi, B., Kumar, R., Manurangsi, P., Pagh, R., and Sinha, A.
\newblock Differentially private aggregation in the shuffle model:almost
  central accuracy in almost a single message.
\newblock In \emph{ICML}, 2021.
\newblock To appear.

\bibitem[Ghosh et~al.(2012)Ghosh, Roughgarden, and
  Sundararajan]{ghosh2012universally}
Ghosh, A., Roughgarden, T., and Sundararajan, M.
\newblock Universally utility-maximizing privacy mechanisms.
\newblock \emph{SICOMP}, 41\penalty0 (6):\penalty0 1673--1693, 2012.

\bibitem[Goryczka \& Xiong(2015)Goryczka and Xiong]{goryczka2015comprehensive}
Goryczka, S. and Xiong, L.
\newblock A comprehensive comparison of multiparty secure additions with
  differential privacy.
\newblock \emph{IEEE TDSC}, 14\penalty0 (5):\penalty0 463--477, 2015.

\bibitem[Greenberg(2016)]{greenberg2016apple}
Greenberg, A.
\newblock {Apple's} ``differential privacy'' is about collecting your data --
  but not your data.
\newblock \emph{Wired, June}, 13, 2016.

\bibitem[Ishai et~al.(2006)Ishai, Kushilevitz, Ostrovsky, and
  Sahai]{ishai2006cryptography}
Ishai, Y., Kushilevitz, E., Ostrovsky, R., and Sahai, A.
\newblock Cryptography from anonymity.
\newblock In \emph{FOCS}, pp.\  239--248, 2006.

\bibitem[Kairouz et~al.(2016)Kairouz, Bonawitz, and Ramage]{KBR16}
Kairouz, P., Bonawitz, K., and Ramage, D.
\newblock Discrete distribution estimation under local privacy.
\newblock In \emph{ICML}, pp.\  2436--2444, 2016.

\bibitem[Kairouz et~al.(2019)Kairouz, McMahan, Avent, Bellet, Bennis, Bhagoji,
  Bonawitz, Charles, Cormode, Cummings, D'Oliveira, Rouayheb, Evans, Gardner,
  Garrett, Gascón, Ghazi, Gibbons, Gruteser, Harchaoui, He, He, Huo,
  Hutchinson, Hsu, Jaggi, Javidi, Joshi, Khodak, Konečný, Korolova,
  Koushanfar, Koyejo, Lepoint, Liu, Mittal, Mohri, Nock, Özgür, Pagh,
  Raykova, Qi, Ramage, Raskar, Song, Song, Stich, Sun, Suresh, Tramèr,
  Vepakomma, Wang, Xiong, Xu, Yang, Yu, Yu, and Zhao]{kairouz2019advances}
Kairouz, P., McMahan, H.~B., Avent, B., Bellet, A., Bennis, M., Bhagoji, A.~N.,
  Bonawitz, K., Charles, Z., Cormode, G., Cummings, R., D'Oliveira, R. G.~L.,
  Rouayheb, S.~E., Evans, D., Gardner, J., Garrett, Z., Gascón, A., Ghazi, B.,
  Gibbons, P.~B., Gruteser, M., Harchaoui, Z., He, C., He, L., Huo, Z.,
  Hutchinson, B., Hsu, J., Jaggi, M., Javidi, T., Joshi, G., Khodak, M.,
  Konečný, J., Korolova, A., Koushanfar, F., Koyejo, S., Lepoint, T., Liu,
  Y., Mittal, P., Mohri, M., Nock, R., Özgür, A., Pagh, R., Raykova, M., Qi,
  H., Ramage, D., Raskar, R., Song, D., Song, W., Stich, S.~U., Sun, Z.,
  Suresh, A.~T., Tramèr, F., Vepakomma, P., Wang, J., Xiong, L., Xu, Z., Yang,
  Q., Yu, F.~X., Yu, H., and Zhao, S.
\newblock Advances and open problems in federated learning.
\newblock \emph{arXiv: 1912.04977}, 2019.

\bibitem[Kasiviswanathan et~al.(2008)Kasiviswanathan, Lee, Nissim,
  Rashkodnikova, and Smith]{kasiviswanathan2008what}
Kasiviswanathan, S.~P., Lee, H.~K., Nissim, K., Rashkodnikova, S., and Smith,
  A.
\newblock What can we learn privately?
\newblock In \emph{FOCS}, pp.\  531--540, 2008.

\bibitem[Kearns(1998)]{kearns1998efficient}
Kearns, M.
\newblock Efficient noise-tolerant learning from statistical queries.
\newblock \emph{JACM}, 45\penalty0 (6):\penalty0 983--1006, 1998.

\bibitem[Knuth(1981)]{Knuth81}
Knuth, D.~E.
\newblock \emph{The Art of Computer Programming, Volume {II:} Seminumerical
  Algorithms, 2nd Edition}.
\newblock Addison-Wesley, 1981.

\bibitem[Kone{\v{c}}n{\`y} et~al.(2016)Kone{\v{c}}n{\`y}, McMahan, Yu,
  Richt{\'a}rik, Suresh, and Bacon]{konevcny2016federated}
Kone{\v{c}}n{\`y}, J., McMahan, H.~B., Yu, F.~X., Richt{\'a}rik, P., Suresh,
  A.~T., and Bacon, D.
\newblock Federated learning: Strategies for improving communication
  efficiency.
\newblock \emph{arXiv: 1610.05492}, 2016.

\bibitem[Kotsogiannis et~al.(2019)Kotsogiannis, Tao, He, Fanaeepour,
  Machanavajjhala, Hay, and Miklau]{kotsogiannis2019privatesql}
Kotsogiannis, I., Tao, Y., He, X., Fanaeepour, M., Machanavajjhala, A., Hay,
  M., and Miklau, G.
\newblock {PrivateSQL}: a differentially private {SQL} query engine.
\newblock \emph{VLDB}, 12\penalty0 (11):\penalty0 1371--1384, 2019.

\bibitem[Kotz et~al.(2001)Kotz, Kozubowski, and Podgorski]{kotz2001book}
Kotz, S., Kozubowski, T., and Podgorski, K.
\newblock \emph{The Laplace Distribution and Generalizations: A Revisit with
  Applications to Communications, Economics, Engineering, and Finance.}
\newblock Progress in Mathematics Series, 2001.

\bibitem[Ruggles et~al.(2019)Ruggles, Flood, Goeken, Grover, Meyer, Pacas, and
  Sobek]{sobek2010integrated}
Ruggles, S., Flood, S., Goeken, R., Grover, J., Meyer, E., Pacas, J., and
  Sobek, M.
\newblock Integrated public use microdata series ({IPUMS}) {USA}: Version 9.0
  [dataset].
\newblock \emph{Minneapolis, MN}, 2019.
\newblock URL \url{https://doi.org/10.18128/D010.V9.0}.

\bibitem[Sason \& Verdu(2016)Sason and Verdu]{sason2016f}
Sason, I. and Verdu, S.
\newblock $f$-divergence inequalities.
\newblock \emph{IEEE TOIT}, 62\penalty0 (11):\penalty0 5973--6006, 2016.

\bibitem[Seide et~al.(2014)Seide, Fu, Droppo, Li, and Yu]{seide20141}
Seide, F., Fu, H., Droppo, J., Li, G., and Yu, D.
\newblock 1-bit stochastic gradient descent and its application to
  data-parallel distributed training of speech {DNNs}.
\newblock In \emph{INTERSPEECH}, pp.\  1058--1062, 2014.

\bibitem[Shankland(2014)]{CNET2014Google}
Shankland, S.
\newblock How {Google} tricks itself to protect {Chrome} user privacy.
\newblock \emph{CNET, October}, 2014.

\bibitem[Stemmer(2020)]{stemmer2020locally}
Stemmer, U.
\newblock Locally private $k$-means clustering.
\newblock In \emph{SODA}, pp.\  548--559, 2020.

\bibitem[Suresh(2019)]{Suresh19}
Suresh, A.~T.
\newblock Differentially private anonymized histograms.
\newblock In \emph{NeurIPS}, pp.\  7969--7979, 2019.

\bibitem[Suresh et~al.(2017)Suresh, Yu, Kumar, and McMahan]{SYKM17}
Suresh, A.~T., Yu, F.~X., Kumar, S., and McMahan, H.~B.
\newblock Distributed mean estimation with limited communication.
\newblock In \emph{ICML}, pp.\  3329--3337, 2017.

\bibitem[Vadhan(2017)]{Vadhan-tutorial}
Vadhan, S.
\newblock \emph{The Complexity of Differential Privacy}, pp.\  347--450.
\newblock Springer International Publishing, 2017.

\bibitem[Wang et~al.(2019)Wang, Xu, Ding, Zhou, Li, and Jha]{wang2019practical}
Wang, T., Xu, M., Ding, B., Zhou, J., Li, N., and Jha, S.
\newblock {MURS}: Practical and robust privacy amplification with multi-party
  differential privacy.
\newblock \emph{arXiv:1908.11515}, 2019.

\bibitem[Warner(1965)]{warner1965randomized}
Warner, S.~L.
\newblock Randomized response: A survey technique for eliminating evasive
  answer bias.
\newblock \emph{JASA}, 60\penalty0 (309):\penalty0 63--69, 1965.

\bibitem[Wilson et~al.(2019)Wilson, Zhang, Lam, Desfontaines, Simmons-Marengo,
  and Gipson]{wilson2019differentially}
Wilson, R.~J., Zhang, C.~Y., Lam, W., Desfontaines, D., Simmons-Marengo, D.,
  and Gipson, B.
\newblock Differentially private {SQL} with bounded user contribution.
\newblock \emph{arXiv:1909.01917}, 2019.

\end{thebibliography}

\appendix


\section{Additional Preliminaries}
\label{app:additional-prelim}

For a real number $y$, we let $y_+ := \max(y, 0)$.  We next recall the hockey stick divergence which belongs to the class of $f$-divergences introduced by \cite{ali1966general, csiszar1964informationstheoretische, csiszar1967information}.

\begin{definition}[Hockey Stick Divergence; e.g., \cite{sason2016f}]\label{def:hockey_stick}
For any $\eps > 0$, the \emph{$e^{\eps}$-hockey stick divergence} between distributions $\cD$ and $\cD'$ is defined as
$d_{\eps}(\cD \| \cD') = \sum_{x \in \supp(\cD)} [\cD(x) - e^{\eps} \cdot \cD'(x)]_+$.
\end{definition}

The following connection between hockey-stick divergence and DP was observed by \cite{barthe2013beyond} and follows immediately from Definitions~\ref{def:dp} and~\ref{def:hockey_stick}.

\begin{lemma} \label{thm:dp-hockey-stick}
A mechanism $\cM$ is $(\eps, \delta)$-DP if and only if $\max_{\bx \sim \bx'} d_{\eps}(\cM(\bx) \| \cM(\bx')) \leq \delta$.
\end{lemma}

\section{Missing proofs from Sections~\ref{sec:dist_mech} and~\ref{sec:corr_dist_mech}}

The following is a well-known fact that will help us determine the privacy guarantee of the $\cD$ Mechanism for $\Delta$-summation; it is an immediate consequence of Lemma~\ref{thm:dp-hockey-stick}.

\begin{lemma} \label{lem:sum-dp-hockey-stick}
For any distribution $\cD$ supported on integers, the $\cD$ Mechanism for $\Delta$-summation is $(\eps, \delta)$-DP if and only if $\max_{-\Delta \leq k \leq \Delta} d_{\eps}(\cD \| k + \cD) \leq \delta$.
\end{lemma}

\subsection{Proof of Theorem~\ref{thm:poisson-dp}}
To prove Theorem~\ref{thm:poisson-dp}, we need the following concentration bound for Poisson distributions (see, e.g.,~\cite{poisson-tail})

\begin{lemma}[Poisson Concentration] \label{lem:concen-poisson}
For any $\lambda, y \in \R^+$,
\begin{align*}
\Pr_{Y \sim \Poi(\lambda)}[|Y - \lambda| \geq y]
\leq 2e^{-\frac{y^2}{2(y + \lambda)}}.
\end{align*}
\end{lemma}

\begin{proof}[Proof of Theorem~\ref{thm:poisson-dp}]
Let $\cD = \Poi(\lambda)$ where $\lambda$ is as specified in the theorem statement.
From Lemma~\ref{lem:sum-dp-hockey-stick}, it suffices to show that $d_{\eps}(\cD \| k + \cD) \leq \delta$ for all $k \in \{-\Delta, \dots, \Delta\}$. To bound $d_{\eps}(\cD \| k + \cD) \leq \delta$, recall that for $\cD = \Poi(\lambda)$, we have $\cD(Y) = \frac{\lambda^Y e^{-\lambda}}{Y!}$. Hence,
\begin{align*}
\frac{\cD(Y)}{\cD(Y - k)} = \lambda^{k} \cdot \frac{(Y - k)!}{Y!}.
\end{align*}
This implies that $\frac{\cD(Y)}{\cD(Y - k)} < e^{|k| \eps}$ for all $Y \in [e^{-\eps/|k|} \lambda + |k|, e^{\eps/|k|} \lambda - |k|]$, which is a superset of $[e^{-\eps / \Delta} \lambda + \Delta, e^{\eps / \Delta} \lambda - \Delta]$. From this, we can bound the hockey stick divergence as follows.
\begin{align*}
\lefteqn{
d_{\eps}(\cD \| k + \cD) = \sum_{Y \in \Z} [\cD(Y) - e^{\eps} \cdot \cD(Y - k)]_+} \\
&= \sum_{Y \in \Z \setminus [e^{-\eps / \Delta} \lambda + \Delta, e^{\eps / \Delta} \lambda - \Delta]} [\cD(Y) - e^{\eps} \cdot \cD(Y - k)]_+ \\
&\leq \sum_{Y \in \Z \setminus [e^{-\eps / \Delta} \lambda + \Delta, e^{\eps / \Delta} \lambda - \Delta]} \cD(Y) \\
&= \Pr_{Y \sim \Poi(\lambda)}[Y < e^{-\eps / \Delta} \lambda + \Delta] + \Pr_{Y \sim \Poi(\lambda)}[Y > e^{\eps / \Delta} \lambda - \Delta].
\end{align*}

From our choice of $\lambda$, we also have $e^{-\eps / \Delta} \lambda + \Delta \leq \lambda - 0.5(1 - e^{-\eps/\Delta})\lambda$ and $e^{\eps/\Delta} \lambda - \Delta \geq \lambda + 0.5(1 - e^{-\eps/\Delta})\lambda$. Hence, we may apply Lemma~\ref{lem:concen-poisson} (with $y = 0.5(1 - e^{-\eps/\Delta})\lambda$) which gives
\begin{align*}
d_{\eps}(\cD \| k + \cD) &\leq 2\exp\left(-\frac{0.25(1 - e^{-\eps/\Delta})^2 \lambda^2}{4 \lambda}\right) \leq \delta,
\end{align*}
where the last inequality follows from the fact that $\lambda \geq 16 \log (2 / \delta) / (1 - e^{-\eps/\Delta})^2$. Hence, we conclude that the $\Poi(\lambda)$ Mechanism is $(\eps, \delta)$-DP as desired.
\end{proof}

\subsection{Proof of Theorem~\ref{thm:nb-privacy}}
To prove Theorem~\ref{thm:nb-privacy}, we need the following generic version of the Chernoff bound; this bound is just the Markov inequality in disguise (note $t < 0$ in the statement below).

\begin{lemma}
\label{lem:chernoff}
For every real number $t < 0$, any real number $a$ and any random variable $X$, $\Pr[X \leq a] \leq \E[e^{tX}]/e^{ta}$.
\end{lemma}

\begin{proof}[Proof of Theorem~\ref{thm:nb-privacy}]
Let $\cD = \NB(r, p)$ where $r, p$ are as specified in the theorem. From Lemma~\ref{lem:sum-dp-hockey-stick}, it suffices to show that $d_{\eps}(\cD \| k + \cD)$ for all $k \in \{-\Delta, \dots, +\Delta\}$. Recall that
\begin{align*}
\cD(y) = \frac{(y + r - 1) \cdots r}{y!} (1 - p)^r p^y,
\end{align*}
for all $y \in \Z_{\geq 0}$. To bound $d_{\eps}(\cD \| k + \cD)$, we consider two separate cases based on the sign of $k$:

\paragraph{Case I: $k \leq 0$.} In this case, we have $\cD(y) \leq \frac{1}{p} \cdot \cD(y + 1) \leq e^{\eps/\Delta} \cdot \cD(y + 1)$ for all $y \in \Z$. From this, we can conclude that $\cD(y) \leq e^{-k \eps / \Delta} \cdot \cD(y - k) \leq e^{\eps} \cdot \cD(y - k)$. As a result, we have $d_{\eps}(\cD \| \cD + k) = 0$. 

\paragraph{Case II: $k > 0$.}
Let $\hy = \Delta - 1 + \frac{r - 1}{e^{1.2\eps/\Delta} - 1}$. If $y \geq \hy$, we have
\begin{align*}
\frac{\cD(y)}{\cD(y - k)}
&= \frac{(y + r - 1) \cdots (y - k + r)}{y \cdots (y - k + 1)} \cdot p^k \\
&\leq \left(\frac{y - k + r}{y - k + 1} \cdot p\right)^k \\
&= \left(\left(1 + \frac{r - 1}{y - k + 1}\right) \cdot p\right)^k \\
&\leq \left(\left(1 + \frac{r - 1}{y - \Delta + 1}\right) \cdot p\right)^{\Delta} \\
(\text{From } y \geq \hy) &\leq \left(e^{1.2\eps/\Delta} \cdot p\right)^{\Delta} \\
&= e^{\eps}.
\end{align*}
As a result, this means that
\begin{align*}
d_{\eps}(\cD \| k + \cD) &= \sum_{y \in \Z} [\cD(y) - e^{\eps} \cdot \cD(y - k)] \\
&= \sum_{y \in \Z \atop y < \hy} [\cD(y) - e^{\eps} \cdot \cD(y - k)] \\
&\leq \Pr_{Y \sim \cD}[Y < \hy].
\end{align*}
It is known that $\E_{X \sim \cD}[e^{tX}] = \left(\frac{1 - p}{1 - pe^t}\right)^r$ for all $t < -\ln p$. As a result, setting $t = -0.2\eps/\Delta$ and applying Lemma~\ref{lem:chernoff}, we have
\begin{align} \label{eq:chernoff}
\Pr_{Y \sim \cD}[Y < \hy] &\leq \frac{\left(\frac{1 - p}{1 - pe^t}\right)^r}{e^{t\hy}}.
\end{align}
Next, using the bound
$e^x \leq 1 + 2x$ for $x \leq 1$ and using $\epsilon \leq 1$, we obtain
\[
e^{1.2\eps/\Delta} - 1
\leq
e^{1.2\Delta} - 1
\leq 2.4/\Delta.
\]
From this, we have
\begin{align*}
\hy &= \Delta - 1 + \frac{r - 1}{e^{1.2\eps/\Delta} - 1} \\
&\leq \frac{2.4 + r - 1}{e^{1.2\eps/\Delta} - 1} \\
&= \frac{r + 1.4}{e^{1.2\eps/\Delta} - 1} \\
&\leq \frac{1.47r}{e^{1.2\eps/\Delta} - 1},
\end{align*}
where the last inequality follows from our choice $r \geq 3$.

Plugging the above inequality back into~\eqref{eq:chernoff}, we get
\begin{align*}
\Pr_{Y \sim \cD}[Y < \hy] &\leq \left(\frac{1 - p}{1 - pe^t} \cdot \exp\left({\frac{-1.47t}{e^{1.2\eps/\Delta} - 1}}\right)\right)^r \\
(\because t = -0.2\eps/\Delta) &\leq \left(\frac{1 - e^{-\frac{0.2\eps}{\Delta}}}{1 - e^{-\frac{0.4\eps}{\Delta}}} \exp\left({\frac{1.47(0.2\eps/\Delta)}{e^{1.2\eps/\Delta} - 1}}\right)\right)^r \\
(\because e^x \geq 1 + x) &\leq \left(\frac{1}{1 + e^{-0.2\eps/\Delta}} \cdot e^{0.25}\right)^r \\
&\leq \left(\frac{e^{0.25}}{1 + e^{-0.2}}\right)^r \\
&\leq e^{-r/3} \\
(\because r \geq 3 \log(1/\delta)) &\leq \delta.
\end{align*}
Thus, $d_{\eps}(\cD \| k + \cD) \leq \delta$ as desired.
\end{proof}

\subsection{Proof of Lemma~\ref{lem:lb-distributed-mechanism}}
\label{app:lb-distributed-mechanism}

To prove Lemma~\ref{lem:lb-distributed-mechanism}, we will resort to Feller's characterization~\cite{Feller} of infinitely divisible distribution as a discrete compound Poisson distribution. We begin by recalling the definition of the latter.

\begin{definition}\label{def:DCP}
A distribution $\cD$ is said to be a \emph{discrete compound Poisson (DCP)} distribution if there exists a distribution $\cD'$ on positive integers and a non-negative real number $\lambda$ such that the following process results in the random variable $Y$ being distributed as $\cD$. First, generate $N \sim \Poi(\lambda)$. Then, let $X_1, \dots, X_N$ be i.i.d. random variables distributed as $\cD'$. Finally, let $Y = X_1 + \cdots + X_N$.

When this condition holds, we write $\cD = \DCP(\lambda, \cD')$.
\end{definition}

A fundamental theorem, due to Feller~\cite{Feller}, states that every infinitely divisible distribution $\cD$ on non-negative integers is a DCP distribution:

\begin{theorem}[\cite{Feller}] \label{thm:inf-div-dcp}
Every infinitely divisible distribution $\cD$ on non-negative integers is a discrete compound Poisson distribution.
\end{theorem}

The above theorem implies the following observation:

\begin{observation} \label{obs:inf-div-mean-vs-var}
For any infinitely divisible distribution $\cD$ on non-negative integers, $\Var(\cD) \geq \E[\cD]$.
\end{observation}

\begin{proof}
From Theorem~\ref{thm:inf-div-dcp}, $\cD = \DCP(\lambda, \cD')$ for some $\lambda \geq 0$ and a distribution $\cD'$ on positive integers. It is well-known that $\E[\cD] = \lambda \cdot \E_{X \sim \cD'}[X]$ and $\Var[\cD] = \lambda \cdot \E_{X \sim \cD'}[X^2]$. Since $\cD'$ is on positive integers, we must have $\E_{X \sim \cD'}[X] \leq \E_{X \sim \cD'}[X^2]$, which implies that $\E[\cD] \leq \Var[\cD]$.
\end{proof}

We are now ready to prove Lemma~\ref{lem:lb-distributed-mechanism}.

\begin{proof}[Proof of Lemma~\ref{lem:lb-distributed-mechanism}]
From Lemma~\ref{lem:sum-dp-hockey-stick}, we have $d_\eps(\cD \| 1 + \cD) \leq \delta$, which can be rearranged as
\begin{align*}
\delta \geq \sum_{i = 0}^{\infty} [\cD(i) - e^{\eps} \cdot \cD(i - 1)]_+,
\end{align*}
which implies that $\cD(i) \leq e^{\eps} \cdot \cD(i - 1) + \delta$ for all non-negative integer $i$. From this and from $\cD$ is supported from non-negative integer (i.e. $\cD(-1) = 0$), it is simple to show via induction that $\cD(i) \leq \delta \cdot \left(\frac{e^{(i + 1)\eps} - 1}{e^{\eps} - 1}\right)$. Let $j = \lfloor \frac{1}{\epsilon} \left(\ln\left(\frac{1}{\delta}\right) - \ln\left(\frac{2}{\eps^2}\right)\right) \rfloor - 1$; we have 
\begin{align*}
\Pr_{Y \sim \cD}[Y < j] &\leq \frac{\delta}{e^{\eps} - 1} \cdot \left(\sum_{i=0}^{j - 1} (e^{(i + 1)\eps} - 1)\right) \\
&\leq \frac{\delta}{e^{\eps} - 1} \cdot \left(\frac{e^{(j + 1)\eps}}{e^{\eps} - 1}\right) \\
&\leq \frac{\delta}{\eps^2} \cdot e^{(j + 1)\eps} \\
&\leq \frac{1}{2},
\end{align*}
where the last inequality follows from our choice of $j$.
Thus, we have $\E[\cD] \geq j \cdot \frac{1}{2} \geq \Omega_\eps(\log(1/\delta))$. Invoking Observations~\ref{obs:inf-div-mean-vs-var} and~\ref{obs:util-simple} immediately yields the desired result.
\end{proof}

\subsection{Proof of Theorem~\ref{thm:dp-neg-noise}}
\begin{proof}[Proof of Theorem~\ref{thm:dp-neg-noise}]
Observe that the analyzer's view only contains $(U_{+1}, U_{-1})$. Since there is a one-to-one correspondence between $(U_{+1}, U_{-1})$ and $(U_{+1} - U_{-1}, U_{-1})$, we may consider the analyzer's view as the latter instead; for convenience, let $U_{\diff} = U_{+1} - U_{-1}$. Note that the distribution of $(U_{\diff}, U_{-1})$ (of course) depends on the input data set $\bx = (x_1, \dots, x_i)$; when we would like to stress this, we write $U_{\diff}^{\bx}, U_{-1}^{\bx}$ instead of the usual $U_{\diff}, U_{-1}$.

Let $Z_i^1, Z_i^2, Z_i^3$ be the random variables that the user samples. Define $\hZ^1, \hZ^2, \hZ^3$ as $\sum_{i \in [n]} Z_i^1, \sum_{i \in [n]} Z_i^2, \sum_{i \in [n]} Z_i^3$ respectively. We have
\begin{align*}
U_{+1} = \left(\sum_{i \in [n]} x_i\right) + \hZ^1 + \hZ^3,
\end{align*}
and
\begin{align}
U_{-1} = \hZ^2 + \hZ^3.
\end{align}
Hence, we also have
\begin{align} \label{eq:u-diff}
U_{\diff} = \left(\sum_{i \in [n]} x_i\right) + \hZ^1 - \hZ^2.
\end{align}
Moreover, from how $Z^1_i, Z^2_i, Z^3_i$'s are sampled, we have that $\hZ^1, \hZ^2, \hZ^3$ are independent random variables with distributions $\cD^1, \cD^2, \cD^3$ respectively.

Next, consider $(x_1, \dots, x_n) = \bx \sim \bx' = (x'_1, \dots, x'_n)$. We can calculate the $e^{\eps_1 + \eps_2}$-hockey stick divergence between $(U_{\diff}^{\bx}, U_{-1}^{\bx}), (U_{\diff}^{\bx'}, U_{-1}^{\bx'})$ as follows:
\begin{align}
&d_{\eps_1 + \eps_2}((U_{\diff}^{\bx}, U_{-1}^{\bx})\|(U_{\diff}^{\bx'}, U_{-1}^{\bx'})) \label{eq:divergence-correlated-noise} \\
&= \sum_{a, b \in \Z} \Big[ \Pr[U_{\diff}^{\bx} = a, U_{-1}^{\bx} = b] \nonumber\\
& \qquad\quad - e^{\eps_1 + \eps_2}
\cdot \Pr[U_{\diff}^{\bx'} = a, U_{-1}^{\bx'} = b] \Big]_+ \nonumber \\
&= \sum_{a, b \in \Z} \Big[\Pr[U_{\diff}^{\bx} = a] \Pr[U_{-1}^{\bx} = b \mid U_{\diff}^{\bx} = a] \nonumber \\
& \qquad\quad - e^{\eps_1 + \eps_2} \cdot \Pr[U_{\diff}^{\bx'} = a] \Pr[U_{-1}^{\bx'} = b \mid U_{\diff}^{\bx'} = a]\Big]_+ \nonumber \\
&= \sum_{a, b \in \Z} \Pr[U_{\diff}^{\bx'} = a] \cdot \Big[\frac{\Pr[U_{\diff}^{\bx} = a]}{\Pr[U_{\diff}^{\bx'} = a]} \cdot \Pr[U_{-1}^{\bx} = b \mid U_{\diff}^{\bx} = a] \nonumber \\
& \qquad\quad - e^{\eps_1 + \eps_2} \cdot \Pr[U_{-1}^{\bx'} = b \mid U_{\diff}^{\bx'} = a]\Big]_+ \nonumber
\end{align}
Now, observe that $\Pr[U_{\diff}^{\bx} = a]$ and $\Pr[U_{\diff}^{\bx'} = a]$ are the probability that the $(\cD_1 - \cD_2)$ mechanism outputs $a$ on input $\bx$ and $\bx'$ respectively. Since we assume that the $(\cD_1 - \cD_2)$ mechanism is $\eps_1$-DP, the ratio between the two must not exceed $e^{\eps_1}$. Plugging this into the above equation, we have
\begin{align}
&d_{\eps_1 + \eps_2}((U_{\diff}^{\bx}, U_{-1}^{\bx})\|(U_{\diff}^{\bx'}, U_{-1}^{\bx'})) \nonumber \\
&\leq \sum_{a, b \in \Z} e^{\eps_1} \cdot \Pr[U_{\diff}^{\bx'} = a] \cdot \Big[\Pr[U_{-1}^{\bx} = b \mid U_{\diff}^{\bx} = a] \nonumber \\
& \qquad\qquad\quad - e^{\eps_2} \cdot \Pr[U_{-1}^{\bx'} = b \mid U_{\diff}^{\bx'} = a]\Big]_+ \nonumber \\
&= \sum_{a \in \Z} e^{\eps_1} \cdot \Pr[U_{\diff}^{\bx'} = a] \cdot \sum_{b \in \Z} \Big[\Pr[U_{-1}^{\bx} = b \mid U_{\diff}^{\bx} = a] \nonumber
\\ & \qquad\qquad\quad - e^{\eps_2} \Pr[U_{-1}^{\bx'} = b \mid U_{\diff}^{\bx'} = a]\Big]_+ \label{eq:hockey-stick-expand}
\end{align}
We now rearrange the term $\Pr[U_{-1}^{\bx} = b \mid U_{\diff}^{\bx} = a]$ as: 
\begin{align*}
& \Pr[U_{-1}^{\bx} = b \mid U_{\diff}^{\bx} = a] \\
&= \sum_{c \in \Zz} \Pr[\hZ^2 = c \wedge U_{-1}^{\bx} = b \mid U_{\diff}^{\bx} = a] \\
&= \sum_{c \in \Zz} \Pr[\hZ^2 = c \wedge \hZ^3 = b - c \mid U_{\diff}^{\bx} = a] \\
 &= \sum_{c \in \Zz} \Pr[\hZ^2 = c \wedge \hZ^3 = b - c \mid U_{\diff}^{\bx} = a] \\
&= \sum_{c \in \Zz} \Pr[\hZ^3 = b - c] \cdot \Pr[\hZ^2 = c \mid U^{\bx}_{\diff} = a] \\
&= \sum_{c \in \Zz} \cD^3(b - c) \cdot \Pr[\hZ^2 = c \mid U^{\bx}_{\diff} = a],
\end{align*}
where the third line follows since $\hZ^3$ is independent of  $(\hZ^2, U_{\diff})$.
Similarly, we have
\begin{align*}
& \Pr[U_{-1}^{\bx'} = b \mid U_{\diff}^{\bx'} = a] \\
&= \sum_{c' \in \Zz} \cD^3(b - c') \cdot \Pr[\hZ^2 = c' \mid U^{\bx'}_{\diff} = a].
\end{align*}
Notice also that $\sum_{c \in \Zz} \Pr[\hZ^2 = c \mid U^{\bx}_{\diff} = a] = \sum_{c' \in \Zz} \Pr[\hZ^2 = c' \mid U^{\bx'}_{\diff} = a] = 1$. Plugging this into the two previous equations, we have
\begin{align*}
& \Pr[U_{-1}^{\bx} = b \mid U_{\diff}^{\bx} = a] 
= \sum_{c, c' \in \Zz} \cD^3(b - c) \\
& \qquad \cdot \Pr[\hZ^2 = c \mid U^{\bx}_{\diff} = a] \Pr[\hZ^2 = c' \mid U^{\bx'}_{\diff} = a],
\end{align*}
and 
\begin{align*}
& \Pr[U_{-1}^{\bx'} = b \mid U_{\diff}^{\bx'} = a] = \sum_{c, c' \in \Zz} \cD^3(b - c') \\
& \qquad \cdot \Pr[\hZ^2 = c \mid U^{\bx}_{\diff} = a] \Pr[\hZ^2 = c' \mid U^{\bx'}_{\diff} = a].
\end{align*}
These imply that
\begin{align*}
&\sum_{b \in \Z} \left[\Pr[U_{-1}^{\bx} = b \mid U_{\diff}^{\bx} = a] - e^{\eps_2} \cdot \Pr[U_{-1}^{\bx'} = b \mid U_{\diff}^{\bx'} = a]\right]_+ \\
&=\sum_{b \in \Z} \Big[\sum_{c, c' \in \Zz} \left(\cD^3(b - c) - e^{\eps_2} \cdot\cD^3(b - c')\right) \\
& \qquad\quad \cdot \Pr[\hZ^2 = c \mid U^{\bx}_{\diff} = a] \Pr[\hZ^2 = c' \mid U^{\bx'}_{\diff} = a] \Big]_+ \\
&\leq \sum_{b \in \Z} \sum_{c, c' \in \Zz} \Big[\left(\cD^3(b - c) - e^{\eps_2} \cdot\cD^3(b - c')\right) \\
& \qquad\quad \cdot \Pr[\hZ^2 = c \mid U^{\bx}_{\diff} = a] \Pr[\hZ^2 = c' \mid U^{\bx'}_{\diff} = a] \Big]_+ \\
&= \sum_{c, c' \in \Zz} \Pr[\hZ^2 = c \mid U^{\bx}_{\diff} = a] \Pr[\hZ^2 = c' \mid U^{\bx'}_{\diff} = a] \\
& \qquad\quad \cdot \sum_{b \in \Z} \left[\cD^3(b - c) - e^{\eps_2} \cdot\cD^3(b - c')\right]_+ \\
&= \sum_{c, c' \in \Zz} \Pr[\hZ^2 = c \mid U^{\bx}_{\diff} = a] \Pr[\hZ^2 = c' \mid U^{\bx'}_{\diff} = a] \\
& \qquad\quad \cdot \sum_{b \in \Z} \left[\cD^3(b) - e^{\eps_2} \cdot\cD^3(b + (c - c'))\right]_+ \\
&= \sum_{c, c' \in \Zz} \Pr[\hZ^2 = c \mid U^{\bx}_{\diff} = a] \Pr[\hZ^2 = c' \mid U^{\bx'}_{\diff} = a] \\
& \qquad\quad \cdot d_{\eps_2}(\cD^3 \| (c - c') + \cD^3).
\end{align*}
Recall our assumption that the $\cD^3$ mechanism for $\Delta$-summation is $(\eps_2, \delta_2)$-DP. From Lemma~\ref{lem:sum-dp-hockey-stick}, we have $d_{\eps_2}(\cD^3 \| (c - c') + \cD^3) \leq \delta_2$ for all $c, c' \in \{0, \dots, \Delta\}$. Thus, we can further bound the above expression as
\begin{align*}
&\sum_{b \in \Z} \left[\Pr[U_{-1}^{\bx} = b \mid U_{\diff}^{\bx} = a] - e^{\eps_2} \cdot \Pr[U_{-1}^{\bx'} = b \mid U_{\diff}^{\bx'} = a]\right]_+ \\
&\leq \delta_2 + \Pr[\hZ^2 > \Delta \mid U^{\bx}_{\diff} = a] + \Pr[\hZ^2 > \Delta \mid U^{\bx'}_{\diff} = a].
\end{align*}
Plugging the above into~\eqref{eq:hockey-stick-expand} yields
\begin{align*}
&d_{\eps_1 + \eps_2}((U_{\diff}^{\bx}, U_{-1}^{\bx})\|(U_{\diff}^{\bx'}, U_{-1}^{\bx'})) \\
&\leq \sum_{a \in \Z} e^{\eps_1} \cdot \Pr[U_{\diff}^{\bx'} = a] \cdot \Big(\delta_2 + \Pr[\hZ^2 > \Delta \mid U^{\bx}_{\diff} = a] \\
& \qquad\qquad\qquad + \Pr[\hZ^2 > \Delta \mid U^{\bx'}_{\diff} = a]\Big) \\
&\leq e^{\eps_1} \cdot \delta_2 + e^{2\eps_1} \cdot \Pr[\hZ^2 > \Delta] + e^{\eps_1} \cdot \Pr[\hZ^2 > \Delta] \\
&\leq e^{\eps_1} \cdot \delta_2 + 2 e^{2\eps_1} \cdot \delta_3,
\end{align*}
where the second inequality follows from $\Pr[U^{\bx'}_{\diff} = a] \leq e^{\eps_1} \cdot \Pr[U^{\bx}_{\diff} = a]$ which in turn holds due to the $\eps_1$-DP of the $(\cD_1 - \cD_2)$ Mechanism, and the last inequality follows directly from the concentration assumption on the noise $\cD_2$. As a result, from Lemma~\ref{thm:dp-hockey-stick}, the $(\cD_1, \cD_2, \cD_3)$-Correlated Distributed Mechanism is $(\eps_1 + \eps_2, e^{\eps_1} \cdot \delta_2 + 2 e^{2\eps_1} \cdot \delta_3)$-DP in the shuffled model as desired.
\end{proof}

\subsection{Proof of Theorem~\ref{th:bin_agg_nearly_one_bit}}

\begin{proof}[Proof of Theorem~\ref{th:bin_agg_nearly_one_bit}]
Let us pick our parameters as follows:
\begin{itemize}
\item $\eps_1 = (1 - \gamma) \eps$ and $\eps_2 = \min\{0.5, \gamma \eps\}$.
\item $\delta_2 = \delta_3 = \frac{\delta}{e^{\eps_1} + 2 e^{2\eps_1}}$
\item $\Delta = \lceil \frac{\log(1/\delta_2)}{\eps_1}  \rceil$.
\item Let $r, p$ be as in Theorem~\ref{thm:nb-privacy} with $\eps = \eps_2, \delta = \delta_2$, i.e.,
\begin{align*} 
p = e^{-0.2\eps_2/\Delta} \text{ and } r = 3\left(1  + \log\left(\frac{1}{\delta_2}\right)\right).
\end{align*}

\end{itemize}
We simply run the protocol $(\cD^1, \cD^2, \cD^3)$-Correlated Distributed Mechanism for $\cD^1 = \cD^2 = \NB(1, e^{-\eps_1})$ and $\cD^3 = \NB(r, p)$. We now argue that this protocol yields the desired privacy, accuracy, and (expected) communication.

\paragraph{Privacy.}
We will apply Theorem~\ref{thm:dp-neg-noise}. To do so, we have to show that the three required conditions are satisfied:
\begin{itemize}[nosep]
\item (Privacy of True Noise) Since $\cD^1 - \cD^2 = \DLap(\eps_1)$, the $(\cD^1 - \cD^2)$ Mechanism is $\eps_1$-DP in the Central model.
\item (Privacy of Correlated Noise) From our choice of $r, p$, Theorem~\ref{thm:nb-privacy} immediately implies that the $\cD^3$ Mechanism is $(\eps_2, \delta_2)$-DP for $\Delta$-summation in the Central model.
\item (Concentration of Negative Noise) We may compute $\Pr_{Y \sim \cD^2}[Y > \Delta]$ as follows:
\begin{align*}
\Pr_{Y \sim \cD^3}[Y > \Delta] &= \sum_{i = \Delta + 1}^{\infty} \cD^2(i) 
= \sum_{i = \Delta+1}^{\infty} (1 - e^{-\eps_1})e^{-\eps_1 i} \\
& = e^{-\eps_1(\Delta + 1)} 
\leq \delta_3,
\end{align*}
where the last inequality follows from our choice of $\Delta$.
\end{itemize}

Hence, we can apply Theorem~\ref{thm:dp-neg-noise} which implies that the $(\cD^1, \cD^2, \cD^3)$-Correlated Distributed Mechanism is $(\eps_1 + \eps_2, e^{\eps_1} \cdot \delta_2 + 2 e^{2\eps_1} \cdot \delta_3)$-DP in the shuffled model. Since $\eps_1 + \eps_2 \leq \eps$ and $e^{\eps_1} \cdot \delta_2 + 2 e^{2\eps_1} \cdot \delta_3 = \delta$, this is indeed the desired privacy guarantee.

\paragraph{Accuracy.}
From Observation~\ref{obs:util-neg}, the MSE of the $(\cD^1, \cD^2, \cD^3)$-Correlated Distributed Mechanism is $\Var(\cD^1) + \Var(\cD^2) = \Var(\cD^1 - \cD^2) = \Var(\DLap((1 - \gamma)\eps_1))$ as desired.

\paragraph{Communication.}
From Observation~\ref{obs:com-neg}, the expected number of messages sent by each user is at most
\begin{align*}
1 + \frac{\E[\cD^1] + \E[\cD^2] + \E[\cD^3]}{n}
&\leq 1 + \frac{e^{\eps_1} + e^{\eps_1} + r / (1 - p)}{n} \\
&\leq 1 + O\left(\Delta \cdot \frac{\log(1/\delta)}{\eps_2 n}\right) \\
&\leq 1 + O\left(\frac{\log(1/\delta)^2}{\eps_1 \eps_2 n}\right) \\
&\leq 1 + O\left(\frac{\log(1/\delta)^2}{\gamma \eps^2 n}\right),
\end{align*}
where the second inequality follows from the choice of $p, r$ and since $\eps \leq O(1)$, the third inequality follows from the choice of $\Delta$, and the last inequality follows from the choice of $\eps_1, \eps_2$ and since $\gamma \leq 1/2$.  
\end{proof}

\section{From Binary Summation to Histograms}\label{sec:red_bit_sum_to_hist}

In this section, we prove Corollary~\ref{th:histograms_single_bit} which we start by recalling:

\begin{corollary}[\bf Histograms]\label{cor:histograms_single_bit_restated}
For every $\epsilon \le O(1)$ and every $\delta, \gamma \in (0, 1/2)$, there is an $(\epsilon, \delta)$-DP protocol for histograms on sets of size $B$ in the multi-message shuffled model, with error equal to a vector of independent discrete Laplace random variables each with parameter $\frac{(1-\gamma)\epsilon}{2}$ and with an expected communication per user of $(\log{B}+1) \cdot ( 1+O(\frac{B\log^2(1/\delta)}{\gamma \epsilon^2 n}))$ bits.
\end{corollary}


\begin{minipage}{0.47\textwidth}
\begin{algorithm}[H]
\caption{Histogram Randomizer.} \label{alg:randomizer_hist}
\begin{algorithmic}[1]
\Procedure{HistogramRandomizer$_{\mathcal{R}}(i)$}{}
\State \textbf{For} $j=1$ \textbf{to} $B$:
\State \quad $S_j \leftarrow \mathcal{R}(\bone[i = j])$
\State \quad $R_j \leftarrow \{j\} \times S_j$
\State \Return $\bigcup_{j=1}^{B} R_j$
\EndProcedure
\end{algorithmic}
\end{algorithm}
\end{minipage}
\hfill
\begin{minipage}{0.48\textwidth}
\begin{algorithm}[H]
\caption{Histogram Analyzer.} \label{alg:analyzer_hist}
\begin{algorithmic}[1]
\Procedure{HistogramAnalyzer$_{\mathcal{A}}(R)$}{}
\State \textbf{For} $j=1$ \textbf{to} $B$:
\State \quad $R_j \leftarrow \{y_1 \; |  \; y \in R \text{ and } y_0 = j\}$
\State \quad $a_j \leftarrow \mathcal{A}(R_j)$
\State \Return $(a_j)_{j=1}^B$
\EndProcedure
\end{algorithmic}
\end{algorithm}
\end{minipage}


\begin{proof}[Proof of Corollary~\ref{cor:histograms_single_bit_restated}]
We define the $(\cD^1, \cD^2, \cD^3)$-Correlated Distributed Histogram Randomizer by instantiating the generic histogram analyzer given in Algorithm~\ref{alg:randomizer_hist} with $\mathcal{R} = ${ \sc Randomizer}$_{\cD^1, \cD^2, \cD^3, n}$ from Algorithm~\ref{alg:randomizer2}. Similarly, we define $(\cD^1, \cD^2, \cD^3)$-Correlated Distributed Histogram Analyzer by instantiating the generic histogram analyzer given in Algorithm~\ref{alg:analyzer_hist} with $\mathcal{A} = ${ \sc Analyzer}$_{\cD^1, \cD^2}$ from Algorithm~\ref{alg:analyzer2}. Thus, the resulting procedures proceed via the same templates as the binary summation randomizer and analyzer in Algorithms~\ref{alg:randomizer2} and~\ref{alg:analyzer2} respectively but by applying it to each of the $B$ buckets. As a consequence, each message consists of an index in $[B]$ in addition to the increment/decrement bit which yields a communication cost per message of $\log{B} + 1$ bits. We next provide the privacy, accuracy, and communication analyses for completeness. To do so, we start by defining the $\infty$-div distributions $\cD^1$, $\cD^2$, and $\cD^3$ to be used in the aforementioned calls to Algorithms~\ref{alg:randomizer_hist} and~\ref{alg:analyzer_hist}. Given $\epsilon$, $\delta$ and $\gamma$ as in Corollary~\ref{cor:histograms_single_bit_restated}, we define the parameters:

\begin{itemize}
\item $\eps_1 = \frac{(1 - \gamma) \eps}{2}$ and $\eps_2 = \frac{\gamma \eps}{2}$.
\item $\delta_2 = \delta_3 = \frac{\delta}{2 \cdot (e^{\eps_1} + 2 e^{2\eps_1})}$
\item $\Delta = \lceil \frac{\log(1/\delta_2)}{\eps_1}  \rceil$.
\item 
Let $r, p$ be as in Theorem~\ref{thm:nb-privacy} with $\eps = \eps_2, \delta = \delta_2$, i.e.,
\begin{align*} 
p = e^{-0.1 \eps_2 / \Delta} \text{ and } r = 50 \cdot e^{\eps_2 / \Delta} \cdot \log\left(\frac{1}{\delta_2}\right).
\end{align*}
\end{itemize}
We set $\cD^1 = \cD^2 = \NB(1, e^{-\eps_1})$ and $\cD^3 = \NB(r, p)$.
Note that these are the same settings as in the proof of Theorem~\ref{th:bin_agg_nearly_one_bit} except that $\epsilon$ is replaced by $\frac{\epsilon}{2}$ and $\delta$ is replaced by $\frac{\delta}{2}$.

\paragraph{Privacy.}
Note that in the shuffled model, the analyzer in Algorithm~\ref{alg:analyzer_hist} only observes the number of $+1$'s and the number of $-1$'s for each of the $B$ buckets. We will first prove the privacy guarantee in the case where the analyzer only observes the number of $+1$'s and the number of $-1$'s  for a \emph{single bucket}, and then extend the argument to the case where it observes all $B$ buckets. For a fixed bucket, the task reduces to the privacy of the bit summation protocol which was shown in the proof of Theorem~\ref{th:bin_agg_nearly_one_bit}. Due to our slightly different setting of parameters where $\epsilon$ and $\delta$ are replaced by $\frac{\epsilon}{2}$ and $\frac{\delta}{2}$ respectively, this implies that the protocol is $(\frac{\epsilon}{2}, \frac{\delta}{2})$-DP. Extending to the case of multiple buckets, we note that any change in a single user's input would affect exactly \emph{two} buckets. Using the fact that the noise variables in these two buckets are independent, the Basic Composition theorem (see, e.g., Theorem B.1 of~\cite{dwork2014algorithmic}), this implies that in the general case where the analyzer observes the counts for all buckets, the histograms protocol given by Algorithms~\ref{alg:randomizer_hist} and~\ref{alg:analyzer_hist} is $(\epsilon, \delta)$-DP as desired.

\paragraph{Accuracy.}
Note that the $B$-dimensional error vector in Algorithm~\ref{alg:analyzer_hist} consists of independent coordinates each equal to the unbiased difference of two independent $\NB(1, e^{-\eps_1})$ random variables. Thus, each coordinate of the error vector is distributed according to $\DLap(\eps_1) = \DLap((1-\gamma) \eps/2)$ as desired.

\paragraph{Communication.}
The expected number of messages sent by each user in Algorithm~\ref{alg:randomizer_hist} is
\begin{align}\label{eq:hist_exp_comm_per_user}
 &1 + B \cdot (\E[\cD^1_{/n}] + \E[\cD^2_{/n}] + 2 \cdot \E[\cD^3_{/n}])\nonumber\\ 
 &= 1 + B \cdot \bigg(\frac{\E[\cD^1] + \E[\cD^2] + 2 \cdot \E[\cD^3]}{n}\bigg)
\end{align}
But $\E[\cD^1] = \E[\cD^2] = \frac{e^{-\eps_1}}{1-e^{-\eps_1}} = \frac{1}{e^{\eps_1}-1} \le \frac{1}{\epsilon_1}$. On the other hand,
\begin{align*}
&\E[\cD^3] = \frac{pr}{1-p} = \frac{e^{-0.1\eps_2/\Delta} \cdot \log(1/\delta_2) \cdot 50 \cdot e^{\eps_2/\Delta}}{1-e^{-0.1\eps_2/\Delta}}\\ 
&= \frac{50 \cdot e^{\eps_2/\Delta} \cdot \log(1/\delta_2)}{e^{0.1\eps_2 / \Delta}-1} \le \frac{500 \cdot \Delta \cdot e^{\eps_2/\Delta} \cdot \log(1/\delta_2)}{\eps_2}.
\end{align*}
Plugging back in Equation~\eqref{eq:hist_exp_comm_per_user} and using the facts that $\Delta = \lceil \frac{\log(1/\delta_2)}{\eps_1}  \rceil$, $\eps_2 = \frac{\gamma \eps}{2}$, $\delta_2 =  \frac{\delta}{2 \cdot (e^{\eps_1} + 2 e^{2\eps_1})}$, and $\eps = O(1)$, we get that the expected number of messages sent by each user is at most
$$ 1 +  B \cdot \bigg( \frac{2}{\eps_1 n} + O\left(\frac{\log^2(1/\delta)}{\eps_1 \eps_2 n}\right) \bigg) = 1 + O\left(\frac{B\log^2(1/\delta)}{\gamma \eps^2 n}\right).$$
The expected number of bits of communication per user stated in Corollary~\ref{cor:histograms_single_bit_restated} now follows using the fact that each message in Algorithm~\ref{alg:randomizer_hist} consists of $\log{B}+1$ bits.
\end{proof}

\paragraph{Improving the Computational Efficiency of Randomizer.}

If implemented naively, the randomizer (Algorithm~\ref{alg:randomizer_hist}) has a running time that depends linearly on $B$, since it has to go over all the buckets $i \in [B]$ and run the randomizer on each bucket, which corresponds to sampling random variables $Z^1_i, Z^2_i, Z^3_i$. Intuitively, this should not be necessary since most of these random variables are zero. Below, we will show that this intuition is indeed correct, in the sense that the randomizer can be significantly sped up, leading to an expected running time that is linear in the expected per-user communication cost.\footnote{For ease of exposition, we assume a model of computation where the logarithm and exponential functions can be computed in constant time, and where a uniform random variable in $(0,1)$ can be sampled in constant time.}

To do so, let us abstract the problem at hand as follows:

\begin{definition}
Let $\cD$ be an infinitely divisible distribution on non-negative integers. In the {\sc ParallelSampling} problem for $\cD$, we are given two positive integers $n, B$ and the goal is to output a multiset $S$ whose elements are from $[B]$ such that, if we let $Y_i$ be the number of occurrences of $i$ in $S$, then $Y_1, \dots, Y_B$ are distributed as i.i.d. $\cD_{/n}$ random variables.
\end{definition}

Recall from Theorem~\ref{thm:inf-div-dcp} that every infinitely divisible distribution on non-negative integers can be represented as a discrete compound Poisson (DCP) distribution. We will exploit this representation to give an efficient algorithm for {\sc ParallelSampling}, as stated below.

\begin{theorem}\label{th:eff_parallel_samp}
Let $\cD$ be the discrete compound Poisson distribution $\DCP(\lambda, \cD')$. Suppose that there is an algorithm $\cA'$ that can sample $\cD'$ in expected time $T$. Then, {\sc ParallelSampling}$_{\cD} (n, B)$ can be performed in expected time $O\left(1+\E[\cD] \cdot \frac{B \log B}{n} + \lambda \cdot T \cdot \frac{B}{n}\right)$.
\end{theorem}

\begin{minipage}{0.47\textwidth}
\begin{algorithm}[H]
\caption{Efficient Sampling for $\cD = \DCP(\lambda, \cD')$.} \label{alg:parallel_sampling}
\begin{algorithmic}[1]
\Procedure{ParallelSampling$_{\cD} (n, B)$}{}
\State $S \leftarrow \{\}$ \label{line:initial}
\State Sample $N_{\text{sum}} \sim \Poi(B \lambda / n)$ \label{line:Poisson_sample}
\State \textbf{For} $\ell = 1$ \textbf{to} $N_{\text{sum}}$:
\State \quad Randomly sample $j \sim [B]$ \label{line:uniform_sample}
\State \quad Use algorithm $\cA'$ to sample $X \sim \cD'$ \label{line:procedure_sample}
\State \quad Add $X$ copies of $j$ to $S$ \label{line:copies_add}
\State \Return $S$
\EndProcedure
\end{algorithmic}
\end{algorithm}
\end{minipage}

We will need the following fact about generating Poisson random variables:
\begin{lemma}[\cite{Knuth81}]\label{th:poisson_sampling}
There is an algorithm that takes a positive real number $\lambda$ and generates a sample from $\Poi(\lambda)$ in expected $O(1+\lambda)$ steps.
\end{lemma}

We are now ready to prove Theorem~\ref{th:eff_parallel_samp}.
\begin{proof}[Proof of Theorem~\ref{th:eff_parallel_samp}]
Our sampling procedure is given in Algorithm~\ref{alg:parallel_sampling}.

\paragraph{Correctness.}
For every $i \in [B]$, let $N_i$ denote the number of iterations for which the index $j$ sampled in line~\ref{line:uniform_sample} is equal to $i$. From standard properties\footnote{Specifically, it is well-known that if we first pick $N_{\text{sum}} \sim \Poi(\tilde{\lambda})$ and let $N_1, \dots, N_n$ be a multinomial distribution with $N_{\text{sum}}$ trials and the probability of incrementing $i$-th bucket in each trial is $p_i$, then $N_i$'s are independent $\Poi(\tilde{\lambda} \cdot p_i)$ random variables. This is sometimes referred to in literature as \emph{Poissonization}.} of the Poisson random variable, we have $N_1, \dots, N_B$ are i.i.d. $\Poi(\lambda / n)$ random variables. As a result, the algorithm is equivalent to: for every $i \in [B]$, independently generate $N_i \sim \Poi(\lambda / n)$, generate $X_1, \dots, X_{N_i} \sim \cD'$ and let $Y_i = X_1 + \cdots + X_{N_i}$. In other words, we have $Y_1, \dots, Y_B$ are i.i.d. $\DCP(\lambda / n, \cD') = \DCP(\lambda, \cD')_{/n} = \cD_{/n}$ as desired.

\paragraph{Expected Running Time.}
We calculate the expected running time for each step of our algorithm:
\begin{itemize}
\item Line~\ref{line:initial} clearly takes constant time.
\item For line~\ref{line:Poisson_sample}, we can use Lemma~\ref{th:poisson_sampling} to sample from $\Poi(B\lambda/n)$ in expected time $O\left(1+\frac{B\lambda}{ n}\right)$.
\item Notice that $\E[N_{\text{sum}}] = B\lambda / n$. Thus, each of the following sub-steps is repeated this many times in expectation. We now list the running time for each sub-step:
\begin{itemize}
\item Line~\ref{line:uniform_sample} takes at most $O(\log B)$ time.
\item From our assumption on $\cA'$, line~\ref{line:procedure_sample} takes $T$ time.
\item Line~\ref{line:copies_add} takes $O(\E[\cD'] \cdot \log B)$ time in expectation.
\end{itemize}
\end{itemize}
Hence, in total, the expected running time of Algorithm~\ref{alg:parallel_sampling} is:
\begin{align*}
&O\left(1+\frac{B \lambda}{n} \cdot \left(\log B + T + \E[\cD'] \cdot \log B\right) \right) \\
&= O\left(1+\frac{B \lambda T}{n} + \lambda \cdot \E[\cD'] \cdot \frac{B \log B}{n}\right) \\
&= O\left(1+\frac{B \lambda T}{n} + \E[\cD] \cdot \frac{B \log B}{n}\right),
\end{align*}
where the first equality used that $\E[\cD'] = \Omega(1)$ which follows from the fact that $\cD'$ is supported on the positive integers. This completes our proof.
\end{proof}

Using our $\cD$-{\sc ParallelSampling} routine and its analysis in Theorem~\ref{th:eff_parallel_samp}, we are now ready to describe our computationally efficient implementation of the histogram randomizer. The pseudo-code is given in Algorithm~\ref{alg:randomizer_hist_efficient}. 

\begin{algorithm}[H]
\caption{Efficient Histogram Randomizer.} \label{alg:randomizer_hist_efficient}
\begin{algorithmic}[1]
\Procedure{HistogramRandomizer$_{\cD^1, \cD^2, \cD^3, n, B}(i)$}{}
\State $S_1 \leftarrow$ {\sc ParallelSampling}$_{\cD^1}(n, B)$
\State $S_2 \leftarrow$ {\sc ParallelSampling}$_{\cD^2}(n, B)$
\State $S_3 \leftarrow$ {\sc ParallelSampling}$_{\cD^3}(n, B)$
\State $R \leftarrow \{(i,+1)\}$
\State \textbf{For} $j \in S_1$:
\State \quad Add $(j,+1)$ to $R$
\State \textbf{For} $j \in S_2$:
\State \quad Add $(j,-1)$ to $R$
\State \textbf{For} $j \in S_3$:
\State \quad Add both $(j,+1)$ and $(j,-1)$ to $R$
\State \Return $R$
\EndProcedure
\end{algorithmic}
\end{algorithm}

Note that the running time of this algorithm is on the same order as the time need to run {\sc ParallelSampling} on distributions $\cD^1, \cD^2, \cD^3$ and inputs $n, B$. Since $\cD^1 = \cD^2 = \NB(1, e^{-\eps_1})$ and $\cD^3 = \NB(r, p)$ in our proof of Corollary~\ref{cor:histograms_single_bit_restated}, our task boils down to using an efficient sampling procedure for the Negative Binomial distribution. Since $\NB(r, p)$ is infinitely divisible, Theorem~\ref{thm:inf-div-dcp} implies that it is a discrete compound Poisson distribution. In this case, the corresponding distribution $\cD'$ on positive integers from Definition~\ref{def:DCP} is known to be the \emph{logarithmic distribution} with parameter~$p$. We will apply Theorem~\ref{th:eff_parallel_samp} with $\cD = \NB(r, p)$ and $\cD'$ set to the logarithmic distribution with parameter~$p$. To do so, we need to upper-bound the time needed to sample from the logarithmic distribution. This is achieved in the following lemma.

\begin{lemma}\label{th:log_dist_sampling}
There is an algorithm that takes a parameter $p \in (0,1)$ and generates a sample from the logarithmic distribution with parameter~$p$ in an expected number of steps proportional to the mean of the distribution.
\end{lemma}

\begin{proof}
Recall that the probability mass function of the logarithmic distribution with parameter $p \in (0,1)$ is given by $f(k) := -\frac{1}{\ln(1-p)} \frac{p^k}{k}$ for any positive integer $k$. We will apply inverse transform sampling. Namely, let $F(\cdot)$ denote the cumulative distribution function. We proceed as follows:
\begin{enumerate}
    \item Sample a uniform number $u$ between $0$ and $1$.
    \item For $k = 2, 3, \dots$,
    \begin{enumerate}
        \item If $F(k) > u$, return $k$.
    \end{enumerate}
\end{enumerate}
Note that this procedure outputs $k$ whenever $u$ lands in the interval $[F(k-1), F(k))$. The probability of this event is equal the width of the interval, which is equal to the probability $f(k)$. Moreover, the expected number of iterations run in the for loop is equal to the mean of the distribution. Note as we iterate over increasing values of $k$ in the above procedure, the value of $F(k)$ can be computed from that of $F(k-1)$ in constant time. This follows from the identity $f(k) = f(k-1) \cdot p \cdot \frac{(k-1)}{k}$ which holds fro all positive integers $k$ larger than $1$.
\end{proof}

By Theorem~\ref{th:eff_parallel_samp} and Lemma~\ref{th:log_dist_sampling}, we conclude that using the implementation in Algorithm~\ref{alg:randomizer_hist_efficient} allows us to reduce the expected running time of the histogram randomizer from $\Omega(B)$ to $O\left(1+\frac{B\log^2(1/\delta)}{\gamma \epsilon^2 n}\right)$ (i.e, to the same order as the expected per user communication complexity).

\section{Lower Bound for Single-Message Binary Summation}\label{sec:lb_bit_sum_single_message}

We now prove a lower bound against single-message protocols for binary summation (Theorem~\ref{th:bin_sum_lb_intro}). 
The lower bound does not depend on any restriction on the size of the message output by the randomizer.
In fact, we prove a slightly stronger statement than in Theorem~\ref{th:bin_sum_lb_intro}, in that our lower bound holds even for $\eps = \ln n - \log n - \Theta(1)$, as stated more precisely below.

\begin{theorem}
\label{th:lb_binary_agg}
Let $\delta = 1/n^c$ and $\epsilon \le \ln{n} - \ln(d \log{n})$ for any positive constants $c$ and  $d$. Then, any $(\epsilon, \delta)$-DP protocol for Binary Summation in the single-message shuffled model must incur an expected squared error of at least $f\log{n}$ for some positive constant $f$ depending on $c$ and $d$.
\end{theorem}

To prove Theorem~\ref{th:lb_binary_agg}, we will use the following result.
\begin{lemma}[\cite{BalleBGN19}, Lemma 4.1]\label{le:sum_analyzer_Balle_et_al}
Let $\cM = (\cR, \cA)$ be an $n$-party protocol for Binary Summation in the single-message shuffled model. Assume that the inputs $X_1, \dots, X_n$ to the $n$ parties are sampled i.i.d. from some distribution on $\{0,1\}$. Then, there is a protocol $\cM' = (\cR', \cA')$ in the single-message shuffled model s.t.:
\begin{enumerate}[nosep]
    \item Each message from the randomizer to the shuffler is a real number in the interval $[0,1]$ and the analyzer simply \emph{sums} its $n$ incoming messages. Namely, $\mathrm{Im}(\cR') \subseteq [0,1]$ and $\cA'(y_1,\dots,y_n) = \sum_{i=1}^n y_i$.
    \item The MSE of the protocol $\cM'$ is at most that of the protocol $\cM$ (with respect to the i.i.d. input distribution). 
    \item If the protocol $\cM$ is $(\epsilon, \delta)$-DP, then the protocol $\cM'$ is $(\epsilon, \delta)$-DP.
\end{enumerate}
\end{lemma}

We are now ready to prove Theorem~\ref{th:lb_binary_agg}.

\begin{proof}[Proof of Theorem~\ref{th:lb_binary_agg}]
Let $\cM = (\cR, \cA)$ be an $(\epsilon, \delta)$-DP protocol for Binary Summation in the single-message shuffled protocol with expected MSE $\alpha$. Consider the case where the inputs of the $n$ users are drawn i.i.d. from the uniform distribution on $\{0,1\}$. By Lemma~\ref{le:sum_analyzer_Balle_et_al}, there is an $(\epsilon, \delta)$-DP protocol $\cM' = (\cR', \cA')$ in the single-message shuffled model where the randomizer outputs a real number in $[0,1]$, the analyzer sums all $n$ incoming messages, and where the MSE of $\cM'$ is at most that of $\cM$. For each $b \in \{0,1\}$, let $\cD_b$ be the probability distribution supported over a subset $T_b$ of the real line and corresponding to the output of the randomizer $\cR'$ on input $b$. Let $T := T_0 \cup T_1$. We next lower bound the total variation distance between $\cD_0$ and $\cD_1$ in terms of~$\alpha$. Specifically, we show that $d_{TV}(\cD_0, \cD_1) \geq  1-\frac{8\alpha}{n}$. Since\footnote{We assume that the messages sent are encoded as bit strings, meaning that $T$ is countable and $\sum_{x \in T} \min\{\cD_0(x), \cD_1(x)\}$ is well-defined.} $d_{TV}(\cD_0, \cD_1) := 1 - \sum_{x \in T} \min\{\cD_0(x), \cD_1(x)\}$, this is equivalent to proving that
\begin{equation}\label{eq:sum_min_ub}
    \sum_{x \in T} \min\{\cD_0(x), \cD_1(x)\} \le \frac{8\alpha}{n}.
\end{equation}
To prove this inequality, note that the MSE of the protocol $\cM'$ satisfies
\begin{align*}
\alpha^2 &= \E\bigg[\bigg(\sum_{i \in [n]} (x_i-y_i)\bigg)^2\bigg]\\ 
&= \E\bigg[\sum_{i \in [n]} \sum_{j \in [n]} (x_i-y_i)(x_j-y_j)\bigg]\\ 
&= \sum_{i \in \in [n]} \E[(x_i-y_i)^2] + \sum_{i \neq j \in [n]} \E[(x_i-y_i)(x_j-y_j)]\\ 
&\stackrel{(*)}{=} n \cdot \E[(x-y)^2] +(n^2-n) \cdot \E[x-y]^2\\ 
&\geq n \cdot \E[(x-y)^2]\\ 
&= \frac{n}{2} \cdot (\E[y^2 ~\mid~ x = 0] + \E[(1-y)^2 ~\mid~ x = 1])\\ 
&=  \frac{n}{2} \cdot \Big( \sum_{y \in T_0} \cD_0(y) y^2 + \sum_{y \in T_1} \cD_1(y) (1-y)^2 \Big) \\ 
&=  \frac{n}{2} \cdot \sum_{y \in T} ( \cD_0(y) y^2 + \cD_1(y) (1-y)^2 )\\  
&\geq \frac{n}{8} \cdot \sum_{y \in T} \min\{\cD_0(y), \cD_1(y)\},
\end{align*}
where the equality (*) uses the i.i.d. property of the user inputs and the local randomizers. Inequality~(\ref{eq:sum_min_ub}) now follows.

Let $\delta = 1/n^c$ for any positive constant $c$, and let $\epsilon$ be any value satisfying
$\epsilon \le \ln{n} - \ln(d \cdot \log{n})$ for some positive constant $d$. We next show that the protocol $\cM'$ violates $(\epsilon, \delta)$-DP whenever $\alpha < f \cdot \log{n}$, where $f$ is a positive constant that depends on $c$ and $d$. To do so, we define $S_0$ to be the set of all values of $y \in T_0$ for which $\cD_0(y) > \cD_1(y)$ and similarly define $S_1$ as the set of all values of $y \in T_1$ for which $\cD_1(y) > \cD_0(y)$. By definition, $S_0$ and $S_1$ are disjoint. Let $\kappa := (8\alpha)/n$. We next argue that $\cD_0(S_0) \geq 1 - \kappa$. By symmetry, it would follow that $\cD_1(S_1) \geq 1 - \kappa$. Assume for the sake of contradiction that $\cD_0(S_0) < 1- \kappa$. Then, $\cD_0(T \setminus S_0) > \kappa$. In this case, we have that
\begin{align*}
\sum_{y \in T} \min\{\cD_0(y), \cD_1(y)\} &\geq \sum_{y \in (T \setminus S_0)} \min\{\cD_0(y), \cD_1(y)\}\\ 
&\geq \sum_{y \in (T \setminus S_0)} \cD_0(y) \\
&= \cD_0(T \setminus S_0)\\ 
&> \kappa \\
&= \frac{8\alpha}{n},
\end{align*}
which would contradict~\eqref{eq:sum_min_ub}.

To show that the protocol $\cM'$ violates $(\epsilon, \delta)$-DP, we consider two neighboring length-$n$ binary input sequences: $\bx = (0,0,\dots,0)$ and $\bx' = (1,0,\dots,0)$. 
To show that $\cM’$ violates $(\epsilon, \delta)$-DP, it suffices to show that:
$$ \Pr[\cR(\bx) \subseteq S_0] > e^{\epsilon}  \cdot \Pr[\cR(\bx') \subseteq S_0] + \delta.$$
We first note that:
\begin{align*}
\Pr[\cR(\bx) \subseteq S_0] & = \cD_0(S_0)^n \\
& \geq (1-\kappa)^n \\
&= \bigg(1-\frac{8 \alpha}{n}\bigg)^n \\
& \geq e^{-8 \alpha -O(\frac{\alpha^2}{n})}.
\end{align*}
On the other hand, we have that:
\begin{align*}
\Pr[\cR(\bx') \subseteq S_0] & = \cD_0(S_0)^{n-1} \cdot \cD_1(S_0) \\
& \le \cD_0(S_0)^{n-1} \cdot \cD_1(T \setminus S_1) \\
& \le \cD_0(S_0)^{n-1} \cdot \kappa.
\end{align*}
Thus,
\begin{align*}
\frac{\Pr[\cR(\bx) \subseteq S_0]}{\Pr[\cR(\bx') \subseteq S_0]} & \geq \frac{\cD_0(S_0)^n} {\cD_0(S_0)^{n-1} \kappa} \\
&= \frac{\cD_0(S_0)}{\kappa} \\
& \geq \frac{1-\kappa}{\kappa} \\
&= \frac{n}{8 \alpha} - 1.
\end{align*}
So it suffices to choose $\alpha$ such that $e^{-8 \cdot \alpha -O(\frac{\alpha^2}{n})} > 2 \delta$ and $\frac{n}{8 \alpha} - 1 > 2 e^{\epsilon}$.
For $\delta = 1/n^c$, where $c$ is a positive constant, and for $\epsilon$ satisfying
$\epsilon \le \ln{n} - \ln(d \log{n})$, where $d$ is a positive constant, we can satisfy the last two inequalities and get a violation of $(\epsilon, \delta)$-DP as long as $\alpha < f \log{n}$ where $f$ is a positive constant depending on $c$ and $d$.
\end{proof}

\section{Experiments for Histogram}
\label{app:exp-histogram}

We next evaluate our algorithms for histogram. Once again, we consider the Correlated Distributed Histogram Mechanism\footnote{This is simply Algorithms~\ref{alg:randomizer_hist},~\ref{alg:analyzer_hist} with $\cR, \cA$ being the randomizers and analyzers of the Correlated Distributed Mechanism.} and the Poisson Histogram Mechanism\footnote{This is simply Algorithms~\ref{alg:randomizer_hist},~\ref{alg:analyzer_hist} with $\cR, \cA$ being the randomizers and analyzers of the Poisson Distributed Mechanism.}.

Apart from our algorithms, we also consider three algorithms, each of which can be considered a generalization of the Randomized Response mechanism in the binary case. Each of the three algorithms is parameterized by a ``noise parameter'' $p$. Below we summarize how they work.
\begin{itemize}
\item \textbf{$B$-Randomized Response ($B$-RR).} In $B$-RR, the randomizer outputs the input bin $x$ with probability $1 - p$. With the remaining probability $p$, it simply outputs a random bin from $1, \dots, B$.
\item \textbf{RAPPOR.} In RAPPOR~\cite{erlingsson2014rappor}, the randomizer first encodes the input $x$ using its one-hot encoding $s$, which is simply a $B$-bit string whose $x$th coordinate is the only one with value 1. Then, the randomizer flips each bit of $s$ independently at random with probability $p$, and output the resulting string $\tilde{s}$.
\item \textbf{Fragmented RAPPOR.} While both $B$-RR and RAPPOR are single-message, Fragmented RAPPOR~\cite{erlingsson2020encode}\footnote{We remark what we are using here is what~\cite{erlingsson2020encode} called ``Attributed-fragmented RAPPOR''. They also introduce another fragmentation technique (called ``report fragmentation''). To the best of our knowledge, however, this other fragmentation technique only helps in terms of privacy in their model when there are multiple shufflers, and hence is irrelevant to our setting.} can be thought of as the multiple-message version of RAPPOR. The randomizer works similarly to RAPPOR, except that, instead of outputting the final string $\tilde{s}$ as a single message, it outputs $B$ messages $(1, \tilde{s}_1), \dots, (B, \tilde{s}_B)$. 

While the naive imeplementation of Fragmented RAPPOR results in the randomizer always outputting $B$ messages,~\cite{erlingsson2020encode} noted that we can instead sends only the coordinates of $\tilde{s}$ that are 1. In this case, the expected number of messages sent per user becomes $(1 - p) + p
(B - 1) = 1 + p(B - 2)$. In other words, the expected number of messages overhead is $p(B - 2)$.

A slightly different algorithm was independently proposed in~\cite{anon-power}. In our experiments, the algorithm of~\cite{anon-power} and Fragmented RAPPOR produce essentially the same results (less than 1\% difference in errors and numbers of messages overhead) and hence we do not specifically discuss the former further here.
\end{itemize}

In all three algorithms, the analyzer just output the unbiased estimator for the count of each bucket.

We remark that, to the best of our knowledge, there is no known efficient algorithm that can accurately compute parameters for the three algorithms above. Due to this, we use ``optimistic'' estimates for the parameters, which means that the errors/messages overhead seen in plots before for them could be better than the true numbers. We stress that, while we use ``optimistic'' estimates for these three algorithms, we use very accurate, ``pessimistic'' estimates for our own algorithms. The detailed explanation on how these parameters are computed is given in Appendix~\ref{app:param-computation}.

Unlike our algorithms, the $\ell_{\infty}$ error of Randomized Respose, RAPPOR and Fragmented RAPPOR are data-dependent. For the experiments, we use two datasets from IPUMS which are available online~\cite{sobek2010integrated}. The remainder of this section is divided based on the two datasets, which will be explained in more details below.

For each selection of parameters $\eps, \delta, n, B$ and a dataset specified below, we run 100 repetitions of each algorithm. We then record the RMSE and the average $\ell_{\infty}$ errors over these repetitions. In all experiments we run below, the errors (both RMSE and $\ell_{\infty}$) of RAPPOR is significantly larger than other algorithms; specifically, RAPPOR's errors are always larger than $2.5\times$ that of RR which is the second-worst algorithm in terms of errors. As a result, we do not include RAPPOR in the plots for readability.

We also remark that, in all the plots below, it will be the case that Fragmented RAPPOR and Poisson are essentially the same, both in terms of errors and expected number of messages overhead, for any reasonably small values of $\delta$. (This is because of the same reason as the binary case where Poisson and binary RR roughly coincide.) Due to this, we will only refer to the Poisson Mechanism in the discussions below; it should be understood that any statement that applies to the Poisson Mechanism also essentially applies to Fragmented RAPPOR. For large values of $\delta$, it turns out Poisson Mechanism is better than Fragmented RAPPOR, both in terms of errors and number of messages overhead (see Figures~\ref{fig:err-city} and~\ref{fig:msg-city}).

\subsection{IPUMS City Dataset.}
The first dataset we consider for histogram computation is the city distribution for the US population in the $1940$ census. This dataset was used for evaluating private histogram algorithms in the shuffled model by \cite{wang2019practical}. As in the previous work, we discard data points with an unidentified city. Doing so, we get $n = 60,313,201$ reports and $B = 915$ cities.

Similar to our experiments in binary summation, we study the effect of varying $\eps, \delta$ on the errors and the number of messages overhead per user. 

\paragraph{Errors.}
The plots for $\ell_\infty$ errors and RMSEs as we vary $\eps$ and $\delta$ are included in Figure~\ref{fig:err-city}. The general trends of the errors for Central, Correlated Distributed and Poisson are very similar to that of the RMSEs in our binary summation experiments (Figure~\ref{fig:err-bin}), with the $B$-RR algorithm always incurring noticably larger errors. One interesting difference between the histogram case here and the binary case is that the gaps between Central/Correlated Distributed and Poisson are smaller for histograms. In fact, when $\delta$ is large (e.g. $\delta \geq 10^{-4}$ for $\eps = 1$), the Poisson Mechanism even incurs slightly lower $\ell_\infty$ errors compared to the Correlated Distributed Mechanism. (See the top-left plot in Figure~\ref{fig:err-bin}.) However, once $\delta$ becomes sufficiently small, the expected behaviour is observed (i.e. Poisson incurs noticeably more error compared to Correlated Distributed).

\begin{figure*}[h!]
\centering
\includegraphics[width=0.46\textwidth]{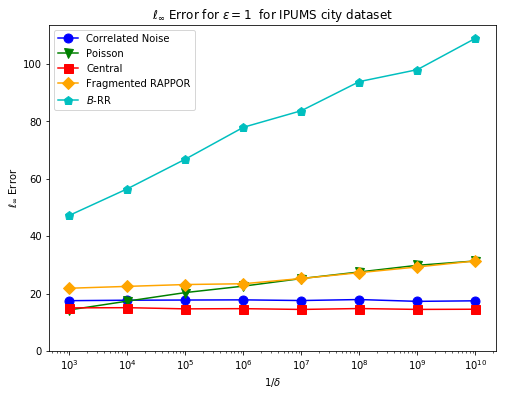}
\includegraphics[width=0.46\textwidth]{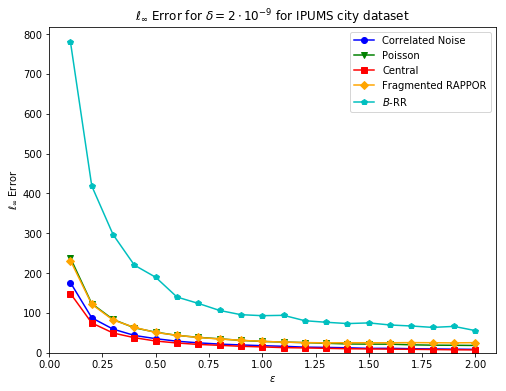}
\includegraphics[width=0.46\textwidth]{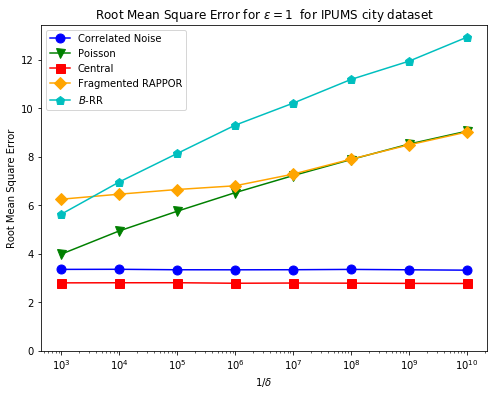}
\includegraphics[width=0.46\textwidth]{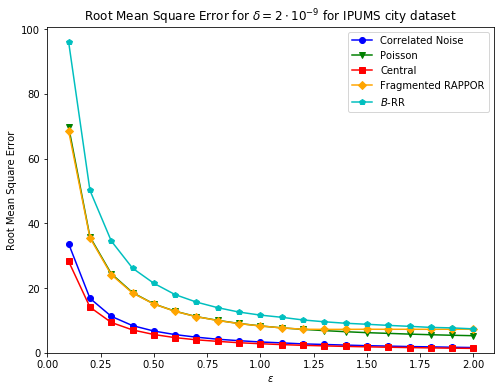}
\caption{RMSEs and $\ell_{\infty}$ errors of different mechanisms on IPUMS city dataset.}
\label{fig:err-city}
\end{figure*}

\paragraph{Communication Complexity.}
We next present our plots for the expected number of additional messages sent for each user in Figure~\ref{fig:msg-city}. The general trends are exactly the same as those shown in Figure~\ref{fig:msg-bin} for binary summation. We also note here that, even in the rather extreme case where $\eps = 0.1$ and $\delta = 2 \cdot 10^{-9}$, the expected number of messages for Correlated Distributed is only 0.181 and for Poisson is only 0.074, both of which are quite small. For moderate value of $\eps = 1$, the numbers dropped to 0.021 and 0.010 respectively.

\begin{figure*}[h!]
\centering
\includegraphics[width=0.46\textwidth]{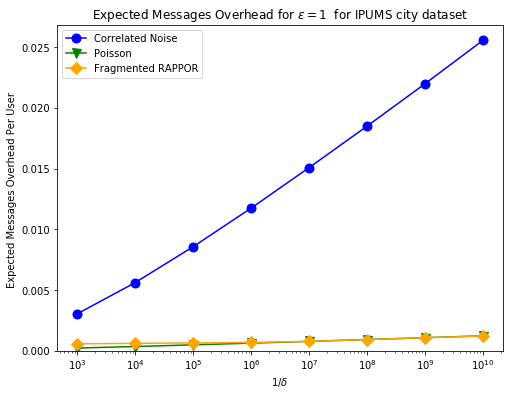}
\includegraphics[width=0.46\textwidth]{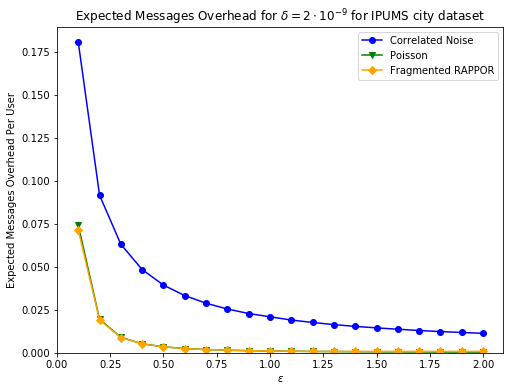}
\caption{Expected additional number of messages sent by each user for IPUMS city dataset}
\label{fig:msg-city}
\end{figure*}

\subsection{IPUMS House Value Dataset.}
Since in the previous dataset the buckets correspond to a categorical feature (namely, the city), it cannot be naturally used to assess the performance of algorithms in terms of a varying number $B$ of buckets. To do so, we also consider the house value distribution of the US population in the 1940 census. After discarding the data points with unavailable house values, we get $n = 14,958,304$ reports with house values between $1$ and $9,750,975$. We then divide the range $[0, 100,000)$ into buckets of equal width and use the width as a knob to vary the number $B$ of buckets in our experiments. More specifically, for a given $B$, we divide the interval $[0, 100,000)$ equally into $B - 1$ buckets, and we put all reports with values at least $100,000$ into the last bucket. (There are only 6009 such reports.)

\paragraph{Errors.} The errors as $B$ varies from 200 to 5000 are shown in Figure~\ref{fig:err-house}. For $\ell_{\infty}$ error, theory predicts that the error for RR increases as $\Omega(B^{1/4})$ in the regime of parameters we consider (see e.g.~\cite{anon-power}) whereas the errors of Central, Correlated Mechanism, and Poisson only increased by $O(\log B)$. This is clearly reflected in the left plot in Figure~\ref{fig:err-house}. On the other hand, the RMSE remains essentially constants as even $B$ varies; see the right plot in Figure~\ref{fig:err-house}.

\paragraph{Communication Complexity.} The number of communication overhead scales linearly in terms of $B$, as shown in Figure~\ref{fig:msg-house}.

\begin{figure*}[h!]
\centering
\includegraphics[width=0.46\textwidth]{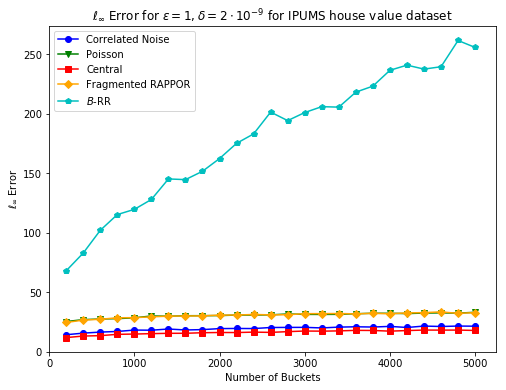}
\includegraphics[width=0.46\textwidth]{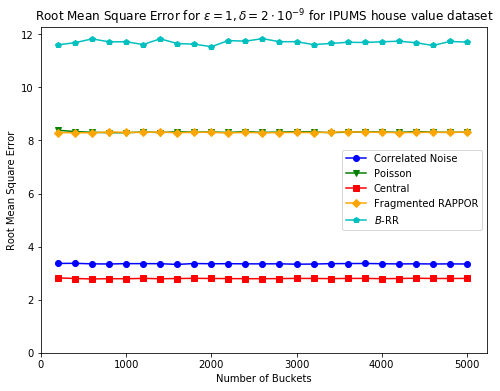}
\caption{RMSEs and $\ell_{\infty}$ errors of different mechanisms on IPUMS house value dataset.}
\label{fig:err-house}
\end{figure*}

\begin{figure*}[h!]
\centering
\includegraphics[width=0.46\textwidth]{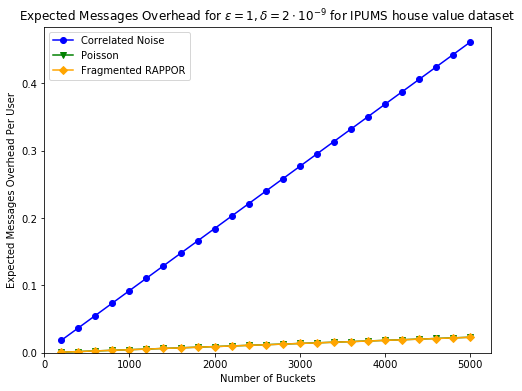}
\caption{Expected additional number of messages sent by each user for IPUMS house value dataset}
\label{fig:msg-house}
\end{figure*}

\section{Parameters Computation for Binary Summation Protocols}
\label{app:bin-param-computation}

\paragraph{Poisson Mechanism.} It is simple to see that the $\Poi(\lambda)$ mechanism becomes more private as $\lambda$ increases\footnote{This just means that if the mechanism was $(\eps, \delta)$-DP, it remains so after $\lambda$ is increased.}. We use binary search to find the smallest $\lambda$ such that the $\Poi(\lambda)$ mechanism is $(\eps, \delta)$-DP. This check can be done efficiently by checking the condition in Lemma~\ref{lem:sum-dp-hockey-stick}. Note that, by selecting the smallest possible $\lambda$, we are minimizing both the MSE and the expected number of messages sent, since $\E[\Poi(\lambda)] = \Var(\Poi(\lambda)) = \lambda$.


\paragraph{Correlated Distributed Noise Mechanism.}
We pick $\cD^1 = \cD^2 = \NB(1, e^{-\eps_1})$ similar to the proof of Theorem~\ref{th:bin_agg_nearly_one_bit}. Here we choose $\eps_1$ such that the root MSE (RMSE) of the protocol is 20\% more than that of the (central) $\DLap(\eps)$ Mechanism.  Once we pick this $\eps_1$ (and hence $\cD^1, \cD^2$), we attempt to find $r, p$ such that, when we set $\cD^3 = \NB(r, p)$, the $(\cD^1, \cD^2, \cD^3)$-Correlated Distributed Mechanism is $(\eps, \delta)$-DP and that it minimizes the expected number of messages sent (or equivalently minimizes $\E[\NB(r, p)] = \frac{rp}{1 - p}$). Now, for a specific $r$, finding the smallest $p = p^*(r)$ such that the Correlated Distributed Mechanism is $(\eps, \delta)$-DP is simple and, similar to Poisson, can be done via a binary search over $p$; note that here we have to compute the expression~\eqref{eq:divergence-correlated-noise}, instead of the expression in Lemma~\ref{lem:sum-dp-hockey-stick}) used for Poisson Mechanism.

On the other hand, we do not know how to efficiently compute $\min_{r \in \R^+} \frac{rp^*(r)}{1 - p^*(r)}$. Hence, we resort to generic optimizers from the \texttt{scipy} package to attempt to find this.  Since we are not guaranteed to find the optimum, it is possible that the optimum number of messages can be even smaller than shown below.

\paragraph{Randomized Response.} Similar to Poisson Mechanism, we use binary search on $p$; for a fixed $p$, we can efficiently check whether the protocol is $(\eps, \delta)$-DP using Lemma~\ref{thm:dp-hockey-stick}.

\section{Parameters Computation for Histogram Protocols}
\label{app:param-computation}

In this section, we detail how the parameters are calculated for the histogram experiments in Appendix~\ref{app:exp-histogram}.

\subsection{Correlated Distributed Histogram Mechanism}

From the proof of Corollary~\ref{cor:histograms_single_bit_restated}, it suffices to select $\cD^1, \cD^2, \cD^3$ so that the $(\cD^1, \cD^2, \cD^3)$-Correlated Distributed Mechanism is $(\eps/2, \delta/2)$-DP.
Similar to the binary summation case, we start by picking $\cD^1 = \cD^2 = \NB(1, e^{-\eps_1})$ where $\eps_1$ is selected so that the RMSE of the protocol is 20\% more than that of $\DLap(\eps/2)$. Once again, we use a similar approach as before to attempt to find $r, p$ such that, when setting $\cD^3 = \NB(r, p)$, the $(\cD^1, \cD^2, \cD^3)$-Correlated Distributed Mechanism is $(\eps/2, \delta/2)$-DP and the expected number of additional messages send is as small as possible.

\subsection{Poisson Histogram Mechanism}

To accurately compute the parameter $\lambda$ needed for the $\Poi(\lambda)$ Histogram Mechanism to be $(\eps, \delta)$-DP. We first prove the following exact characterization of this condition; note that here we use $\cD \otimes \cD'$ to denote the product distribution between $\cD$ and $\cD'$.

\begin{lemma} \label{lem:poisson-hist-dp}
For any $\eps, \delta, \lambda \geq 0$, the $\Poi(\lambda)$ Histogram Mechanism is $(\eps, \delta)$-DP in the shuffled model if and only if
\begin{align}
d_{\eps}(\Poi(\lambda) \otimes \Poi(\lambda) \| (1 + \Poi(\lambda)) \otimes (-1 + \Poi(\lambda))) \leq \delta
\label{eq:dp-poisson-histogram}
\end{align}
\end{lemma}

\begin{proof}
Let $f_{hist}: [B]^n \to (\N \cup \{0\})^B$ denote the function that computes a histogram, i.e., $f_{hist}(x_1, \dots, x_n) = (\sum_{j \in [n]} \bone[x_j = i])_{i \in [B]}$. In a similar vein as Observation~\ref{obs:dp-simple}, it is simple to see that the view (after shuffling) of the $\Poi(\lambda)$ Histogram Mechanism is exactly the same as if we apply the ($B$-dimensional) central $\Poi(\lambda)$ Mechanism to $f_{hist}$. For brevity, let $\cM$ denote the latter mechanism.

To calculate the DP parameters for the mechanism, consider two neighboring data sets $\bx = (x_1, \dots, x_n)$ and $\bx' = (x_1, \dots, x_{n - 1}, x'_n)$. Due to symmetry, we may assume that $x_n = 1$ and $x'_n = 2$. For notational convenience, let $\ba$ denote the $f_{hist}(\bx)$, i.e. the histogram constructed from input $\bx$. We have 
\begin{align*}
&d_{\eps}(\cM(\bx) \| \cM(\bx')) \\
&= \sum_{\by \in \Z^B} \left[\Pr[\cM(\bx) = \ba + \by] - e^\eps \Pr[\cM(\bx') = \ba + \by]\right]_+ \\
&= \sum_{\by \in \Z^B} [\Pr[Y_1 = y_1, \dots, Y_{B} = y_{B}] - \\ 
&\quad e^\eps \Pr[Y_1 = y_1 + 1, Y_2 = y_2 - 1, Y_3 = y_3, \dots, Y_B = y_B]]_+.
\end{align*}
Now, since $Y_1, \dots, Y_B$ are independent $\Poi(\lambda)$ random variables, we can further rewrite the above as
\begin{align*}
&\sum_{\by \in \Z^B}\Pr[Y_3 = y_3, \dots, Y_B = y_B] \cdot [\Pr[Y_1 = y_1, Y_2 = y_2] \\
&\qquad  \qquad \qquad \qquad - e^\eps \Pr[Y_1 = y_1 + 1, Y_2 = y_2 - 1]]_+  \\
&= \sum_{y_1, y_2 \in \Z} [\Pr[Y_1 = y_1, Y_2 = y_2] \\
& \qquad - e^\eps \Pr[Y_1 = y_1 + 1, Y_2 = y_2 - 1]]_+  \\
&= d_{\eps}\left(\Poi(\lambda) \otimes \Poi(\lambda) \| (1 + \Poi(\lambda)) \otimes (-1 + \Poi(\lambda))\right).
\end{align*}
From Lemma~\ref{thm:dp-hockey-stick}, we have completed our proof.
\end{proof}

With the above lemma in mind, we simply binary search on $\lambda$ with the check being condition~\eqref{eq:dp-poisson-histogram}. Once again, we remark that the left hand side of~\eqref{eq:dp-poisson-histogram} can be computed to arbitrary accuracy quite efficiently.

\subsection{Fragmented RAPPOR}
\label{subsec:fragmented-rappor-param-computation}

Unlike the parameters for our own protocols, we will use ``optimistic'' parameters for the remaining protocols, which means that the errors and expected numbers of messages overhead shown in our plots might be smaller than the actual values. For Fragmented RAPPOR, it is not hard to see that the algorithm becomes more private as $p$ increases for any $p \in [0, 1/2]$. Hence, we may use binary search to find the noise parameter~$p$. To do so, we only need a subroutine that, for a given $p$, decides whether the mechanism is $(\eps, \delta)$-DP in the shuffled model. Our algorithm will be ``optimistic''. In the sense that, if the algorithm says NO, then the mechanism is not $(\eps, \delta)$-DP. However, if our algorithm says YES, then the mechanism may still \emph{not} be $(\eps, \delta)$-DP.

The checking task above further reduce to the following: given $p, \eps$, find a lower bound on $\delta$. Specifically, we can prove a lemma similar to Lemma~\ref{lem:poisson-hist-dp}; the main difference is that the implication is one-way instead of two-way (``if and only if'') as in Lemma~\ref{lem:poisson-hist-dp}. However, this weaker implication is already sufficient to give a lower bound on $\delta$.

\begin{lemma}
For any $p \in (0, 1)$ and any positive integer $n$, let $\cD_{n, p}$ denote the distribution $\Ber(1 - p) + \Bin(n - 1, p)$, where $\Ber(1 - p)$ denote the Bernoulli distribution with success probability $1 - p$. For any $\eps, \delta > 0$, if Fragmented RAPPOR is $(\eps, \delta)$-DP in the shuffled model for $n$ users and $B \geq 3$ buckets, then
\begin{align} \label{eq:dp-fragmented-rappor}
d_{\eps}(\cD_{n, p} \otimes \Bin(n, p) \| \Bin(n, p) \otimes \cD_{n, p}) \leq \delta.
\end{align}
\end{lemma}

The proof is very similar to Lemma~\ref{lem:poisson-hist-dp} with $\bx = (B, \dots, B, 1)$ and $\bx' = (B, \cdots, B, 2)$, and when we restrict to only the first two coordinates. We do not repeat the full argument here. We only note that the difference in quantifier comes because here we choose $\bx, \bx'$ ourselves, whereas the previous proof of Lemma~\ref{lem:poisson-hist-dp} works for any neighboring $\bx, \bx'$.


We note that the left hand side in~\eqref{eq:dp-fragmented-rappor} can be computed in $O(n^2)$ by writing it as
\begin{align} \label{eq:dp-fragmented-rappor-expanded}
&d_{\eps}(\cD_{n, p} \otimes \Bin(n, p) \| \Bin(n, p) \otimes \cD_{n, p}) \\
&= \sum_{i=0}^n \sum_{j=0}^n [\cD_{n, p}(i) \Bin(j; n, p) - e^{\eps} \cdot \Bin(i; n, p) \cD_{n, p}(j)]_+,
\end{align}
where we use $\Bin(i; n, p)$ to denote the probability mass of $\Bin(n, p)$ at $i$. We can now compute $d_{\eps}(\cD_{n, p} \otimes \Bin(n, p) \| \Bin(n, p) \otimes \cD_{n, p})$ in $O(n^2)$ time, by enumerating $i, j \in \{0, \dots, n\}$ and compute the inner term. (Note that we can precompute factorials so that computing each $\Bin(\cdot ; n, p), \cD_{n, p}(\cdot)$ takes only $O(1)$ time.)

In our settings of parameters (e.g. where $n \geq 6 \times 10^7$ in the case of city dataset), this $O(n^2)$ algorithm is too slow. To overcome this, we apply pruning techniques. Specifically, instead of considering all $i \in \{0, \dots, n\}$, we only consider $i \in [p n - \tau, p n + \tau]$ for some small number $\tau$. Since we are only dropping non-negative terms, this still gives us a valid lower bound on $\delta$. Furthermore, due to standard concentration bounds, it can easily be seen that taking $\tau = O(\sqrt{n} \cdot  \poly\log(n/\delta))$ suffices to compute the sum~\eqref{eq:dp-fragmented-rappor-expanded} to within an error of say 0.0001$\delta$. By applying a similar prunning to $j$, we end up with an algorithm that runs in time $O(n \cdot \poly\log(n/\delta))$, which suffices for our purposes.

\subsection{$B$-Randomized Response and RAPPOR}

For $B$-RR, it is simple to see that the mechanism becomes more private as $p$ increases for any $p \in [0, 1]$. The same holds for RAPPOR for $p \in [0, 1/2]$.  
Once again, we use a similar approach as in the previous subsection to give optimistic estimates for the noise parameters of $B$-RR and RAPPOR.
As discussed in Appendix~\ref{app:exp-histogram}, even these optimistic errors are already noticably larger than those of the other protocols considered.

Recall form the previous subsection that we only need to provide the following subroutine: given $p, \eps$, find a lower bound on $\delta$. 
We will describe a generic algorithm for such a task in the next two subsections. After that, we will describe how to initiate the algorithms specifically for $B$-Randomized Response and RAPPOR.

\subsubsection{Lower Bound on $\delta$ for Single Message Randomizers}

In this subsection, we give a generic lower bound on $\delta$ given $\eps$ for any single-message mechanism $\cM$ in the shuffled model. Suppose that the set of input of each user is $[B]$. For every $b \in [B]$, we use $\cD_b$ to denote the distribution of the output message of the randomizer when the input is $b$.

To derive this lower bound, we consider two neighboring databases $\bx = (1, \cdots, 1, x)$ and $\bx' = (1, \cdots, 1, x')$ where $x \neq x'$ are from $\{2, \dots, B\}$. We then compute the $\eps$-hockey stick divergence of $\cM(\bx)$ and $\cM(\bx')$. Due to Lemma~\ref{thm:dp-hockey-stick}, this is a lower bound on $\delta$ for which $\cM$ is $(\eps, \delta)$-DP. 

An advantage in taking vectors $\bx, \bx'$ as specified above is that their hockey stick divergence turns out to have a reasonably simple formula:

\begin{lemma} \label{lem:lower-bound-single}
For $\bx = (1, \dots, 1, x)$ and $\bx' = (1, \dots, 1, x')$, we have 
\begin{align} \label{eq:lb-hockey-stick}
d_{\eps}(\cM(\bx) || \cM(\bx')) = \frac{1}{n} \E\left[\left[\sum_{i=1}^n U_i\right]_+\right],
\end{align}
where $U_1, \dots, U_n$ are i.i.d. random variables where $U_i$ is sampled as follows. Sample an outcome (i.e. a possible output message) $o$ according to the distribution $\cD_1$ and let
\begin{align*}
U_i = \frac{\cD_x(o) - e^{\eps} \cdot \cD_{x'}(o)}{\cD_1(o)}.
\end{align*}
\end{lemma} 

The proof of Lemma~\ref{lem:lower-bound-single} is essentially the same as that of Lemma 5.3 from~\cite{BalleBGN19}, except that we replace the ``blanket distribution'' there with the output distribution with 1 as an input $\cD_1$. Nonetheless, we give the full proof for completeness below.

\begin{proof}[Proof of Lemma~\ref{lem:lower-bound-single}]
Throughout the proof, we use $\bone$ to denote the all-zero vector.

Consider any multiset of messages $\bo = \{o_1, \dots, o_n\}$. The probability that this multiset $\bo$ is the output of $\cM(\bx)$ can be rearranged as follows:
\begin{align*}
&\Pr[\cM(\bx) = \bo] \\
&= \sum_{\pi: [n] \to [n]} \cD_1(o_{\pi(1)}) \cdots \cD_1(o_{\pi(n - 1)}) \cD_x(o_{\pi(n)}) \\
&= \sum_{\pi: [n] \to [n]} \cD_1(o_1) \cdots \cD_1(o_n) \frac{\cD_x(o_{\pi(n)})}{\cD_1(o_{\pi(n)})} \\
&= n! \cdot \cD_1(o_1) \cdots \cD_1(o_n) \cdot \left(\frac{1}{n} \sum_{i=1}^n \frac{\cD_x(o_i)}{\cD_1(o_i)}\right) \\
&= \Pr[\cM(\bone) = \bo] \cdot \left(\frac{1}{n} \sum_{i=1}^n \frac{\cD_x(o_i)}{\cD_1(o_i)}\right).
\end{align*}
Similarly, we have
\begin{align*}
\Pr[\cM(\bx') = \bo] = \Pr[\cM(\bone) = \bo] \cdot \left(\frac{1}{n} \sum_{i=1}^n \frac{\cD_{x'}(o_i)}{\cD_1(o_i)}\right).
\end{align*}
As a result,
\begin{align*}
&d_{\eps}(\cM(\bx) || \cM(\bx')) \\
&= \sum_{\bo} \left[\Pr[\cM(\bx) = \bo] - e^{\eps} \cdot \Pr[\cM(\bx') = \bo] \right]_+ \\
&= \sum_{\bo} \frac{\Pr[\cM(\bone) = \bo]}{n} \cdot \left[\sum_{i=1}^n \frac{\cD_{x}(o_i) - e^{\eps} \cdot \cD_{x'}(o_i)}{\cD_1(o_i)}\right]_+ \\
&= \E_{\bo \sim \cM(\bone)}\left[\sum_{\bo} \frac{1}{n} \cdot \left[\sum_{i=1}^n \frac{\cD_{x}(o_i) - e^{\eps} \cdot \cD_{x'}(o_i)}{\cD_1(o_i)}\right]_+\right] \\
&= \frac{1}{n} \E_{o_1, \cdots, o_n \sim \cD_1}\left[\left[\sum_{i=1}^n \frac{\cD_{x}(o_i) - e^{\eps} \cdot \cD_{x'}(o_i)}{\cD_1(o_i)}\right]_+\right] \\
&= \frac{1}{n} \E\left[\left[\sum_{i=1}^n U_i\right]_+\right]. \qedhere
\end{align*}
\end{proof}

\subsubsection{Efficiently Approximating Expected Positive Part of Sum of i.i.d. Random Variables}
\label{subsec:nonneg-iid-sum}

Our task now becomes how to compute the right hand side term of~\eqref{eq:lb-hockey-stick}. While this seems hard in general, the $U_i$'s that we will use below have very specific structures. Specifically, its support is only of size three; that is,
\begin{align*}
U_i =
\begin{cases}
d_1 & \text{ with probability } p_1, \\
d_2 & \text{ with probability } p_2, \\
d_3 & \text{ with probability } p_3.
\end{cases}
\end{align*}
It turns out that, in this case, the problem is trivially solvable in $O(n^2)$ time, because the desired non-negative part of sum $\E\left[\left[\sum_{i=1}^n U_i\right]_+\right]$ can be written in the following form:
\begin{align*}
\sum_{a_1=0}^n \Bin(a_1; n, p_1) \cdot \sum_{a_2=0}^{n - a_1} \Bin\left(a_2; n - a_1, \frac{p_2}{p_2 + p_3}\right) \cdot
\nonumber \\ [a_1d_1 + a_2d_2 + (n - a_1 - a_2)d_3]_+. 
\end{align*}


Although the trivial computation of the above expression requires $O(n^2)$ time, we can apply pruning techniques similar to those in Section~\ref{subsec:fragmented-rappor-param-computation} to reduce the running time to $O(n \cdot \poly\log(n/\delta))$, which suffices for our purposes.

\subsubsection{Parameters for $B$-Randomized Response}

For $B$-RR, we do not use Lemma~\ref{lem:lower-bound-single} directly on the messages, but we first group the messages into three groups (1) the message $x$, (2) the message $x'$ and (3) all other messages. In other words, we ``merge'' all messages that are neither $x$ nor $x'$ into a single outcome. Since merging messages does not increase\footnote{This is because such a merging is a form of post-processing.} hockey-stick divergence, we may compute the hockey stick divergence between $\cM(\bx)$ and $\cM(\bx')$ after such a merging to get a lower bound on $\delta$. In this case, it is not hard to check that we have
\begin{align*}
U_i = 
\begin{cases}
\frac{B}{p} \left((1 - p + \frac{p}{B}) - e^\eps \cdot \frac{p}{B}\right) & \text{ with probability } \frac{p}{B}, \\
\frac{B}{p} \left(\frac{p}{B} - e^\eps \cdot (1 - p + \frac{p}{B})\right) & \text{ with probability } \frac{p}{B}, \\
\frac{p(B - 2)}{B - 2p} \cdot (1 - e^{\eps}) & \text{ w.p. } 1 - \frac{2p}{B}, \\
\end{cases}
\end{align*}
where the first value corresponds to when the output is $x$, the second to when the output is $x'$, and the third to when the output is neither $x$ nor $x'$.

We use the algorithm from the previous subsection to compute a lower bound on $\E\left[\left[\sum_{i=1}^n U_i\right]_+\right]$, which from Lemma~\ref{lem:lower-bound-single} gives us a lower bound on $\delta$.

\subsubsection{Parameters for RAPPOR}

Similarly, for RAPPOR, we first group the messages into three types: (i) the string has the $x$th bit set to one, the $x'$th bit set to zero and the first bit set to zero, (ii) the string has the $x'$th bit set to one, the $x$th bit set to zero and the first bit set to zero, and (ii) all other messages. We remark here that the grouping here is necessary as otherwise there are 7 possible values that the random variable $U_i$ can take, which would result in an algorithm that is too slow (even after pruning). With this grouping, it is simple to check that our $U_i$ has the following distribution:
\begin{align*}
U_i =
\begin{cases}
\left(\frac{1 - p}{p}\right)^2 - e^{\eps} & \text{ w.p. } p^2(1 - p), \\
1 - e^{\eps} \left(\frac{1 - p}{p}\right)^2  & \text{ w.p. } p^2(1 - p), \\
\left(1 - e^{\eps}\right) \cdot \frac{1 - (1 - p)^3 - p^2(1 - p)}{1 - 2p^2(1- p)} & \text{ w.p. } 1 - 2p^2(1- p), \\
\end{cases}
\end{align*}
where the three values corresponding to the three groups respectively.

Once again, we then use the algorithm from Subsection~\ref{subsec:nonneg-iid-sum} to compute a lower bound on $\E\left[\left[\sum_{i=1}^n U_i\right]_+\right]$, which from Lemma~\ref{lem:lower-bound-single} gives us a lower bound on $\delta$.

\end{document}